\documentclass[accepted]{uai2025} % after acceptance, for a revised version; 
\usepackage[american]{babel}
\usepackage{natbib} % has a nice set of citation styles and commands
\usepackage{mathtools} % amsmath with fixes and additions
\usepackage{booktabs} % commands to create good-looking tables
\usepackage{tikz} % nice language for creating drawings and diagrams

\usepackage{latexsym}
\usepackage[utf8]{inputenc}
\usepackage{amssymb}
\usepackage{bm}
\usepackage{dsfont}
\usepackage{nccmath}
\usepackage{amsmath}
\usepackage{mathtools}
\usepackage{mathrsfs}
\usepackage{amsthm}
\usepackage{amsfonts}
\usepackage{algorithm}
\usepackage{algorithmic}
\usepackage{multirow}
\usepackage{diagbox}
\usepackage{comment}
\usepackage{color}
\usepackage{subfigure}
\usepackage{caption}
\usepackage{enumitem}
\usepackage{mathtools}
\usepackage{physics}
\usepackage{optidef}

\newtheorem{assumption}{Assumption}
\newtheorem{definition}{Definition}
\newtheorem{lemma}{Lemma}
\newtheorem{theorem}{Theorem}
\newtheorem{proposition}{Proposition}

\newtheorem{problem}{Problem}
\newtheorem{remark}{Remark}
\newtheorem{example}{Example}

\DeclareMathOperator*{\argmax}{arg\,max}

\title{Probability-Raising Causality\\ for Uncertain Parametric Markov Decision Processes with PAC Guarantees}

\author[1]{\href{mailto:<jj@example.edu>?Subject=Your UAI 2025 paper}{Ryohei~Oura}{}}
\author[2]{Yuji~Ito}
\affil[1]{%
    Frontier Research Center, Toyota Motor Corporation\\
    Japan
}
\affil[2]{%
    Toyota Central R\&D Labs., Inc.\\
    Japan
}

\begin{document}
\maketitle
\begin{abstract}
  Recent decision-making systems are increasingly complicated, making it crucial to verify and understand their behavior for a given specification. A promising approach is to comprehensively explain undesired behavior in the systems modeled by Markov decision processes (MDPs) through formal verification and causal reasoning. However, the reliable explanation using model-based probabilistic causal analysis has not been explored when the MDP's transition probabilities are uncertain.
  This paper proposes a method to identify potential causes of undesired behaviors in an uncertain parametric MDP (upMDP) using parameter sampling, model checking, and a set covering for the samples. A cause is defined as a subset of states based on a \textit{probability-raising} principle.
  We show that the probability of each identified subset being a cause exceeds a specified threshold. Further, a lower bound of the probability that the undesired paths visit the subsets is maximized as much as possible while satisfying a nonredundancy condition.
  While computing these probabilities is complicated, this study derives probabilistically approximately correct lower bounds of both probabilities by the sampling. We demonstrate the effectiveness of the proposed method through a path-planning scenario.
\end{abstract}

\section{Introduction}\label{sec:intro}
% 1. 形式検証により確率的かつ複雑なシステム（ここでMDPに触れる）の振る舞いの正しさを保証するプロセスは重要であり，要求違反の説明や原因の特定のために因果解析がpromissingなアプローチである．
% システムの形式検証は重要であるが，なぜそのような要求違反が生じるのかについて十分な情報は提供しないこと．
As modern decision-making systems have become more complex and exhibit various forms of stochasticity in many applications such as traffic networks and robotic delivery systems \citep{dudek1996taxonomy, dorri2018multi}, it is increasingly important to understand their intricate behavior to answer \ \textit{why} they violate or satisfy a given specification. 
For stochastic systems like Markov decision processes (MDPs), model checking-based formal verification has proven effective in determining whether these systems satisfy the specifications \citep{baier2008principles}. However, model checking typically provides only basic information: a satisfaction guarantee or a counterexample that violates the specification, rather than an in-depth explanation of the system's behavior.

% 2．因果推論に基づく，要求違反の要因を抽出する試みがこれまでなされていること（反事実による因果，ゲームによる因果，マルコフモデル上の因果などについて引用）．
% 近年，MCやMDP上での要求違反の確率上昇に基づく因果関係が導入されており有望であるが，モデル自身の不確かさを考慮できていない．
% 一般的にexactなMDPを入手することは容易ではないこと，および，そのために実用上モデル自身が不確かな状況における因果解析が重要であることを述べる．モデル自身が不確かな場合のチャレンジとなる部分を簡単に述べる．
As a promising approach to explain such behaviors, cause-effect reasoning has attracted a lot of attention in both machine learning and formal method communities \citep{baier2021verification, oberst2019counterfactual, baier2022operational, coenen2022temporal, finkbeiner2023counterfactuals, baier2024foundations, finkbeiner2024synthesis, dimitrova2020probabilistic, kazemi2022causal, baier2024backward, oberst2019counterfactual, tsirtsis2021counterfactual, triantafyllou2022actual, baier2021game} (see Section \ref{RelatedWorks:CausalAnalysis} for more details). Recently, probabilistic causality for MDPs based on a \textit{probability-raising} (PR) principle has been advocated \citep{baier2021probabilistic, baier2022probability, baier2024foundations}. PR causes are defined as a subset of states of an MDP such that passing through the states increases the likelihood of undesired behavior compared to bypassing them, regardless of action selections. The authors have proposed a model checking-based algorithm to intensively identify ones that satisfy a path coverage criterion, offering an explainability advantage in verifying MDPs over other causal reasoning studies. 
% \cite{baier2021probabilistic} have defined the probabilistic causality for paths on Markov chains (MCs) built on the probability-raising and counterfactual reasoning principles. \cite{baier2022probability, baier2024foundations} have extended this concept to the subsets of states of MDPs as probability-raising causes (PR causes) by incorporating nondeterminism and presented a model checking-based algorithm to obtain PR causes that maximize an introduced path coverage metric. 

% 3. 不確かなMDPとしてupMDPが知られており，モデルパラメータのサンプリングに基づきupMDP上の解析に確率的な保証を与えるシナリオ最適化（サンプリングアプローチ）が上記のチャレンジにアタックする上で有望である．
% シナリオアプローチのメリットを述べる．
While this PR-causality framework has been established well for exact MDPs, reliably computing PR causes becomes a challenge when the transition probabilities themselves are epistemically uncertain. Although the algorithms in \citep{baier2022probability, baier2024foundations} have relied on fully known MDPs, this assumption is often unrealistic in complicated systems. Common approaches to handle the epistemically uncertain parameters of MDPs, such as transition probabilities, are using intervals \citep{givan2000bounded, delahaye2011decision, puggelli2013polynomial, jackson2020safety} and defining a class of uncertain (parametric) MDPs (upMDPs) \citep{nilim2005robust, wiesemann2013robust, wolff2012robust, rickard2024learning}. 
Furthermore, \textit{scenario optimization program} (SOP) \citep{campi2008exact, campi2011sampling} has emerged as a powerful approach for handling such parameter uncertainty via sampling. 

Although SOP-based formal verification and synthesis in upMDPs well-investigated in \citep{cubuktepe2020scenario, badings2022scenario, nilim2005robust, wiesemann2013robust, wolff2012robust, rickard2024learning} (see Section \ref{RelatedWorks:VerificationSynthesis} for details), the following non-trivial question remains in the causal reasoning context: Identifying \textit{potential} PR causes \textit{exhaustively} but \textit{nonredundantly}. 
\cite{cubuktepe2020scenario, badings2022scenario}, which are most related to this paper, have provided probabilistically approximately correct (PAC) lower bounds for satisfaction probabilities by applying model checking to sampled MDPs and SOP formulations, but they do not offer causative states.
% Probabilistically guaranteed formal verification combined with sampling parameters have intensively studied for the upMDP \cite{}. The authors in \cite{} have proposed to obtain the satisfaction probability 

% Although the existing works to verify and synthesize upMDPs are well-investigated, the following non-trivial question remains in the causal reasoning context: Identifying \textit{potential} PR causes \textit{exhaustively} but \textit{nonredundantly}. 

% 4. 本論文では，上記の難しさを克服するために，シナリオアプローチ（サンプリングアプローチ）とモデル検査法による，確率的な保証付きの，upMDP上の原因状態集合を推定する方法を提案する．
In this paper, based on parameter sampling, model checking, and a set covering strategy for the samples, we propose a method to exhaustively and nonredundantly identify subsets of states that act as potential PR causes in a upMDP.
% , where the parameters follow unknown distributions from which we can sample the parameter values. 
% Different than existing verification and synthesis studies for upMDPs using SOP formulations \citep{cubuktepe2020scenario, badings2022scenario, badings2022sampling, nilim2005robust, wiesemann2013robust, wolff2012robust, rickard2024learning}, we focus on probabilistic causal analysis to identify the responsible states over a upMDP. 
We show that the resulting collection of subsets satisfies the following three properties in a PAC manner: (i) Potential cause: Each identified subset of states becomes a PR cause with at least a specified probability. (ii) Exhaustiveness: A lower bound of the probability that the undesired paths traverse the collection as much as possible is maximized.
(iii) Nonredundancy: Removing a subset of states from the collection decreases the lower bound. 

% 5. 本論文の貢献は以下である．
% ① upMDP上の確率因果の推定問題を定式化したこと
% ② モデルサンプリングとモデル検査法に基づく効率的な確率上昇因果の推定法を提案したこと
% ③ 任意の状態集合に対して，その状態集合がupMDP上で原因となる確率の下限を確率的な保証付きで与えたこと
% ④ さらに，提案アルゴリズムにより得られる状態の集合族に対して，各集合が一定の確率以上で原因でありつつ，recallの意味で最適であることを示したこと．
% ⑤ いくつかのシミュレーションにより，提案法の効果を示していること．
Our contributions are summarized as follows:
\begin{enumerate}
    % \item We formalize the problem of estimating the PR causes on a upMDP in a PAC manner, thereby connecting probabilistic causal analysis and upMDPs. 
    \item We formalize the problem of exhaustively and nonredundantly identifying PR causes on a upMDP by introducing two probabilities that capture (a) the likelihood of a subset of states being a PR cause and (b) the likelihood that undesired paths maximally pass through such subsets. This addresses the limitation in \citep{baier2022probability} when applying to uncertain models.
    \item We provide lower bounds for the introduced probabilities over the parameter space based on a state set inclusion condition and a path condition. (Theorems \ref{thm_certif} and \ref{thm_optimality}). This result enables the computation of the PAC-guaranteed probabilities using sampled parameters.
    % facilitating a reliable identification of potential PR causes in a upMDP.
    % \item We propose an algorithm to identify the potential PR causes while calculating the PAC-guaranteed lower bounds derived in 2 based on parameter sampling, model checking, and a set covering for the samples (Algorithms \ref{Alg_PG_SPREst} and \ref{Alg_SinCheck}).
    % \item We show that the obtained collection satisfies the aforementioned properties (i), (ii), and (iii) (Theorem \ref{thm_solution_certif}).
    \item We develop a set covering-based method combined with model checking-based PR cause synthesis to compute a solution to the problem formalized in 1 with PAC-guarantees, and its correctness is established in Theorem \ref{thm_solution_certif}. This result paves the way towards compact and verifiable explanations under aleatoric and epistemic uncertainties.
    \item We demonstrate the effectiveness of the proposed method with a path-planning scenario. The proposed method effectively finds exhaustive and nonredundant key waypoints that substantially influence undesired outcomes. However, two na\"ive approaches to be compared do not. This case study suggests a potential utility in control applications with epistemic and aleatoric uncertainties.
\end{enumerate}

% \begin{notations}
% \label{notations}
 % For any set $A$ and $B$, we denote by $A \subseteq B$ (resp., $A \subset B$) when $A$ is a subset (resp., proper subset) of $B$.
% \end{notations}

\section{Preliminaries}
\subsection{Uncertain Parametric Markov Decision Processes}
We first consider a parametric Markov decision process (pMDP). Then, we extend it toward an uncertain parametric MDP (upMDP) with probability distributions over the parameter space of the pMDP.
% from which we can sample the parameter values.

A pMDP is a tuple $\mathcal{M}_u = (S, A, s_0, u, P)$, where $S$ is a finite set of states; $A$ is a finite set of actions; $s_0 \in S$ is an initial state; $u = [p_0, \ldots, p_{m-1}]^\top$ is a real-valued parameter vector, where $u$ takes a value in a subset of $\mathbb{R}^m$ denoted by $\mathcal{V} \subseteq \mathbb{R}^m$; and $P: S \times A \times S \times \mathcal{V} \to [0,1]$ is a transition function, where $P(s,a,s',u)$ denotes the probability of transition from $s$ to $s'$ with the action $a$. For any state $s \in S$, we denote the set of enabled actions at $s$ by $\mathcal{A}(s) = \{ a \in A \;|\; \exists u, \exists s' \mbox{ s.t. } P(s,a,s',u) > 0 \}$. We say that $s \in S$ is terminal if $\mathcal{A}(s) = \emptyset$. We denote the set of terminal states by $E$ and assume $E \neq \emptyset$. A path is a sequence of states and actions $\rho = s_0 a_0 s_1 a_1 s_2 \ldots $ such that $P(s_i, a_i, s_{i+1}, u) > 0$ for some $u \in \mathcal{V}$ and all indices $i \geq 0$. 
% We say that a path $\rho$ is maximal if $\rho$ is infinite or finite and ends in a terminal state. 
A policy $\pi : \mathrm{FinPath} \times A \to [0,1]$ is a function that assigns a probability to a pair of a finite path and an enabled action at the last state, where $\mathrm{FinPath}$ is the set of all finite paths.
% We say that a pMDP $\mathcal{M}_u$ is \textit{graph-preserving} if, for any $s, s' \in S$ and any $a \in A$, there are no $u, u' \in \mathcal{V}$ such that $P(s, a, s', u) \neq 0$ and $P(s, a, s', u') = 0$. We assume that any pMDP is graph-preserving.
A path reaching $E$ is interpreted as an undesired one.

% Note that we can yield an MDP $\mathcal{M}_u$ from a pMDP $\mathcal{M}$ by assigning an instantiation $u \in \mathcal{V}$  
A upMDP is a tuple $\mathcal{M}_\mathbb{P} = (\mathcal{M}_u, \mathbb{P})$, where $\mathcal{M}_u$ is a pMDP and $\mathbb{P}$ is a probability distribution over the parameter space $\mathcal{V}$ of $\mathcal{M}$. Intuitively, a upMDP is a pMDP combined with an associated probability distribution over possible parameter values. So, a parameter $u \in \mathcal{V}$ generated from $\mathbb{P}$ yields a concrete MDP $\mathcal{M}_u$ if the instantiated transition function $P_u = P(\cdot, \cdot, \cdot, u)$ is well-defined, that is, $\sum_{s' \in S} P_u(s,a,s') = 1$ for any $s \in S$ and any $a \in A$ if $a \in \mathcal{A}(s)$, otherwise $0$. We assume that all parameter values in $\mathcal{V}$ yield well-defined MDPs. 
For any MDP $\mathcal{M}_u$, any state $s \in S$, and any policy $\pi$, we define the probability measure $\mathrm{Pr}^\pi_{\mathcal{M}_u, s}$ on measurable sets of paths starting from $s$ under $\pi$ in a standard way \citep{baier2008principles}. For simplicity, we abbreviate $\mathrm{Pr}^\pi_{\mathcal{M}_u, s_0}$ as $\mathrm{Pr}^\pi_{\mathcal{M}_u}$.
% For any measurable set $T$ of maximal paths, we denote by $\mathrm{Pr}^\mathrm{min}_{\mathcal{M}_u,s}(T)$ (resp., $\mathrm{Pr}^\mathrm{max}_{\mathcal{M}_u,s}(T)$) the infimum (resp., supremum) of the probabilities of $T$ under all policies.

\begin{example}
\label{example1}
We consider a upMDP depicted in Fig. \ref{example3:upMDP} with the set of parameter variables $u = [p, q]^\top$, the set of actions $A=\{a, b\}$, and the terminal state $E = \{s_5\}$. The enabled action is only $a$ at states other than $s_0$. Their transition probabilities are represented using parameters $p$ and $q$. $p$ and $q$ are equal to $1/2$ with probability $1/10$ and, with probability $9/10$, $p$ and $q$ follow the uniform distributions over $[0.11, 0.51]$ and $[0.3, 0.7]$, respectively.
\begin{figure}[htbp]
\centering
    \includegraphics[width=0.7\linewidth]{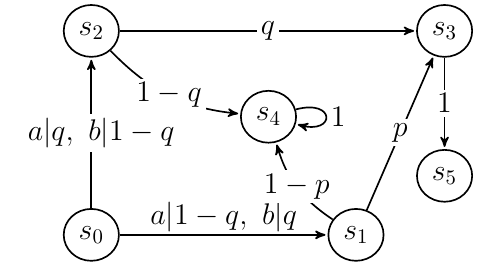}
\caption{The upMDP with the parameter variables $V=\{p, q\}$, the set of actions $A=\{ a, b\}$, and the terminal state $E = \{ s_5 \}$. The enabled action is only $a$ at states other than $s_0$ and thus we omit it.}
\label{example3:upMDP}
\end{figure}
\end{example}

\subsection{Probability-Raising Causality}
For convenience, we use linear temporal logic (LTL)-like notations such as $\neg$ (negation), $\land$ (and),  $\Diamond$ (eventually), and $\mathrm{U}$ (until) to represent some path properties. For any $X_0, X_1 \subseteq S$, the formula $X_0 \mathrm{U} X_1$ is satisfied by a (state) path $\rho = s_0 s_1 \ldots$ such that there exists $j \geq 0$ such that, for any $i < j$, $s_i \in X_0$ and $s_j \in X_1$ hold, and we denote $\Diamond X_1 = S \mathrm{U} X_1$. Furthermore, $\neg \Diamond X_1$ is satisfied by a path $\rho$ such that all states in $\rho$ are not contained by $X_1$.
% and $\Diamond T \land \neg \Diamond T'$ is satisfied by a path $\rho = s_0 s_1 \ldots$ such that there exists $s_i \in T$ but all states in $\rho$ are not contained by $T'$.
For simplicity, we abbreviate $\{s_0\} \mathrm{U} \{ s_1\}$ and $\Diamond \{s_2\}$ as $s_0 \mathrm{U} s_1$ and $\Diamond s_2$, respectively, for any $s_0, s_1, $ and $s_2 \in S$.

We consider a cause-effect relation over an MDP based on the probability-raising principle introduced in \citep{baier2022probability, baier2024foundations}.
\begin{definition}[SPR cause \citep{baier2022probability}]
\label{def:SPR_cause}
    For any MDP $\mathcal{M}_u$ with a parameter $u \in \mathcal{V}$, the set of terminal states $E \subseteq S$, and a subset $C \subseteq S \setminus E$, we say that $C$ is a strict probability-raising (SPR) cause for $E$ on $\mathcal{M}_u$ if and only if the following conditions (M) and (S) hold:
    \begin{description}
        \item[(M)] $\forall c \in C, \exists \pi$ such that $\mathrm{Pr}_{\mathcal{M}_u}^\pi(\neg C \mathrm{U} c) > 0$.
        \item[(S)] $\forall c \in C, \forall \pi \mbox{ where } \mathrm{Pr}_{\mathcal{M}_u}^\pi (\neg C \mathrm{U} c) > 0$,
        \begin{align}
        \mathrm{Pr}_{\mathcal{M}_u}^\pi(\Diamond E \;|\; \neg C \mathrm{U} c) > \mathrm{Pr}_{\mathcal{M}_u}^\pi(\Diamond E).
        \end{align}
    \end{description}
\end{definition}
Intuitively, $\Diamond E$ means that a path eventually reaches a state in $E$, and $\neg C \mathrm{U} c$ denotes that a path eventually reaches $c$ avoiding others in $C$.
SPR cause is a minimal subset $C$ of states and is interpreted as waypoints that raise the probability of reaching terminal states $E$. This probability-raising condition is defined by $\textbf{(S)}$, implying that visiting any state in $C$ increases the probability of reaching $E$. The minimality of the subset is defined by $\textbf{(M)}$, removing states in $C$ that can only be reached by passing through the others.
% The condition \textbf{(M)} can be viewed as a minimality requirement for a set of states $C$ to omit any state that is not reachable from the initial state without going through others in $C$. 
We denote that $C$ satisfies the above conditions \textbf{(M)} and \textbf{(S)} on an MDP $\mathcal{M}_u$ by $(\mathcal{M}_u, C) \models \textbf{(M)}$ and $(\mathcal{M}_u, C) \models \textbf{(S)}$, respectively.
% Furthermore, we denote by $(\mathcal{M}_\mathbb{P}, C) \models \textbf{(M)}$ when $C$ satisfies the above condition \textbf{(M)} for every $\mathcal{M}_u$ with $u \in \mathcal{V}$ such that $p(u) > 0$, where $p$ is the probability density function for $\mathbb{P}$. Since $\mathcal{M}_\mathbb{P}$ is graph-preserving, for any $C \subseteq S$, we have $(\mathcal{M}_\mathbb{P}, C) \models \textbf{(M)}$ when $(\mathcal{M}_u, C)$ for some $u \in \mathcal{V}$.

In Example \ref{example1}, the set of states $C=\{s_2, s_3\}$ is an SPR cause on the MDP $\mathcal{M}_{u}$ with $p=0.3$ and $q=0.6$ for the following reasons: First, there exist two paths $s_0 a s_2$ and $s_0 b s_1 a s_3$. Hence, $C$ satisfies the condition \textbf{(M)} on $\mathcal{M}_u$. Second, the minimal probabilities over any policies $\mathrm{Pr}^\mathrm{min}_{\mathcal{M}_u}(\Diamond s_5 \;|\; \neg C \mathrm{U} s_2) = 0.6$ and $\mathrm{Pr}^\mathrm{min}_{\mathcal{M}_u}(\Diamond s_5 \;|\; \neg C \mathrm{U} s_3) = 1$ are greater than the maximal probability $\mathrm{Pr}^\mathrm{max}_{\mathcal{M}_u}(\Diamond s_3) = 0.48$. Thus, $C$ satisfies \textbf{(S)} on $\mathcal{M}$.

In the following, we do not mention the set of terminal states $E$ for any upMDP when it is clear from the context.
\begin{definition}[SPR-cause probability]
    For any upMDP $\mathcal{M}_\mathbb{P}$ and any $C \subset S$, we define the probability that $C$ is an SPR cause on $\mathcal{M}_\mathbb{P}$ as:
\begin{align}
\label{cause_prob}
    F(\mathcal{M}_\mathbb{P}, C) = \int_{u \in \mathcal{V}} I(u, C) \mathrm{d} \mathbb{P}(u),
\end{align}
where $I: \mathcal{V} \times 2^S \to \{ 0, 1 \}$ is an indicator function such that, for any $u \in \mathcal{V}$ and any $C \subset S$, $I(u, C) = 1$ if $C$ is an SPR cause on $\mathcal{M}_u$ and otherwise $0$. We call the probability defined by (\ref{cause_prob}) the \textit{SPR-cause probability} for $C$.
\end{definition}

% 直感的な/要約的な(?) 定義の振り返り．
% cause自体の最適性を考えたいというモチベーションを書く．その後に，「既存研究の～～という定義を採用する，さらにそれをupMDPに拡張する」のように導入していく．
To evaluate how often an SPR cause covers paths reaching $E$ on the upMDP, we consider the probability of paths passing through an SPR cause $C$ conditioned on reaching $E$ as a probabilistic coverage of SPR causes.
% Moreover, assessing the coverage by ignoring some rare SPR causes is useful since such SPR causes can be seen as outliers for $\mathbb{P}$. 
\begin{definition}[Recall]
    For any MDP $\mathcal{M}_u$, any policy $\pi$, and any $C \subseteq S$, we introduce the conditional probability called recall for $C$ under $\pi$ as
\begin{align}
\label{def:recall}
    \mathrm{Pr}_{\mathcal{M}_u}^\pi (\Diamond C \;|\; \Diamond E),
\end{align}
and $\mathrm{Pr}^\pi_{\mathcal{M}_u}(\Diamond C \;|\; \Diamond E) = 0$ if $\mathrm{Pr}^\pi_{\mathcal{M}_u}(\Diamond E) = 0$.
\end{definition}
% Intuitively, the probability defined as (\ref{def:recall}) is the probability that the paths reaching $E$ go through $C$.
\begin{definition}[Recall optimality]
\label{def:recall_optimality}
    For any MDP $\mathcal{M}_u$, any $C \subseteq S$, and any $\mathcal{S} \subseteq 2^S$, we say that $C$ is recall-optimal on $\mathcal{M}_u$ over $\mathcal{S}$ if and only if $C$ is an SPR cause and satisfies
    \begin{align}
    \label{PR-recall_optimal}
        & \forall C' \in \mathcal{S}, C' \mbox{ is an SPR cause on $\mathcal{M}_u$} \nonumber \\
        & \implies \min_{\pi} \mathrm{Pr}_{\mathcal{M}_u}^\pi (\Diamond C | \Diamond E) - \mathrm{Pr}_{\mathcal{M}_u}^\pi (\Diamond C' | \Diamond E) \!\geq\! 0,
    \end{align}
\end{definition}
Intuitively, when the conditional probability of passing through $C$ before making undesired paths is maximal among all SPR causes, $C$ is recall optimal.
\begin{definition}[Recall-optimal probability]
\label{def:recall_optimal_probability}
    For any MDP $\mathcal{M}_u$ and any $\mathcal{C}, \mathcal{S} \in 2^S$, we define the probability that $\mathcal{C} \subseteq 2^S$ covers recall-optimal SPR causes over $\mathcal{S} \subseteq 2^S$ under $\mathcal{M}_\mathbb{P}$ as:
    \begin{align}
    \label{recall_opt_prob}
        R(\mathcal{M}_\mathbb{P}, \mathcal{C}, \mathcal{S}) = \int_{u \in \mathcal{V}} I'(u, \mathcal{C}, \mathcal{S}) \mathrm{d} \mathbb{P}(u),
    \end{align}
    where $I': \mathcal{V} \times 2^{2^S} \times 2^{2^S} \to \{ 0, 1 \}$ is an indicator function such that, for any $u \in \mathcal{V}$, any $\mathcal{C} \subseteq 2^S$, and any $\mathcal{S} \subseteq 2^S$, $I'(u, \mathcal{C}, \mathcal{S}) = 1$ if there exists $C \in \mathcal{C}$ such that $C$ is recall-optimal over $\mathcal{S}$ under $\mathcal{M}_u$ and otherwise $0$. We call the probability defined by (\ref{recall_opt_prob}) the \textit{recall-optimal probability} for $\mathcal{C}$ over $\mathcal{S}$.
\end{definition}
Recall optimal probability is interpreted as the probability over the parameter space of a upMDP that a given collection of subsets of states contains a recall optimal SPR cause.
% That is, the probability that the collection has an SPR cause $C$ to achieve the maximal conditional probability of passing through $C$ before reaching terminal states.

Moreover, we define a collection of subsets of states whose SPR-cause probabilities exceed $\delta > 0$ as 
\begin{align}
\label{S_delta}
    \mathcal{S}_\delta = \{ C \subset S \;|\; F(\mathcal{M}_\mathbb{P}, C) > \delta \}.
\end{align}
% In the following, we mainly treat $\mathcal{S}$ in Defs. \ref{def:recall_optimality} and \ref{def:recall_optimal_probability} as $\mathcal{S}_\delta$.

%% Problem Setting
\section{Problem Setting}

% The goal of this paper is to find a nonredundant but exhaustive collection of subsets of states $\mathcal{C}^*$ such that $\mathcal{C}^*$ maximizes the recall-optimal probability while the SPR-cause probability for each $C \in \mathcal{C}^*$ exceeds a certain threshold on the given upMDP. Formally, we consider the following problem.
Our goal is to find a nonredundant but exhaustive collection of potential SPR causes $\mathcal{C}^* \subset 2^S$ on the given upMDP. To reflect this intuition, we consider the following problem.
\begin{problem}
\label{problem:exact}
    For any $\delta$, compute $\mathcal{C}^* \subset 2^S$ that satisfies the following two conditions:
    \begin{description}
        \item[(C1)] $\mathcal{C}^* \in \argmax_{\mathcal{C} \subseteq \mathcal{S}_\delta} R(\mathcal{M}_\mathbb{P}, \mathcal{C}, \mathcal{S}_\delta)$, 
        \item[(C2)] $\forall \mathcal{C}' \subset \mathcal{C}^*, R(\mathcal{M}_\mathbb{P}, \mathcal{C}^*, \mathcal{S}_\delta) > R(\mathcal{M}_\mathbb{P}, \mathcal{C}', \mathcal{S}_\delta)$,
    \end{description}
    where $\mathcal{S}_\delta$ is defined by (\ref{S_delta}).
\end{problem}
In Problem \ref{problem:exact}, \textbf{(C1)} and \textbf{(C2)} state an exhaustiveness and a nonredundancy conditions for $\mathcal{C}^*$, respectively, for the recall-optimal probability.
Intuitively, \textbf{(C1)} requires that $\mathcal{C}^*$ maximally covers the paths reaching some terminal states in a probabilistic sense while the SPR-cause probability for each $C \in \mathcal{C}^*$ exceeds a given threshold $\delta$. The maximality of recall-optimal probability corresponds to exhaustiveness, and the threshold for the SPR-cause probability reflects our intuition for potential causes. Plus, \textbf{(C2)} requires that the collection $\mathcal{C}^*$ is nonredundant in the sense that removing any member of $\mathcal{C}^*$ reduces the recall-optimal probability.
We can analyticaly compute $F$ defined by (\ref{cause_prob}) and $R$ defined by (\ref{recall_opt_prob}) if $\mathcal{V}$ is finite and $\mathbb{P}$ is known.
However, in practice, $\mathcal{V}$ can be infinite, and $\mathbb{P}$ can be unknown. Hence, it is generally undecidable to compute the probabilities \citep{arming2018parameter}. Thus, we relax Problem \ref{problem:exact} by a probably approximately correct (PAC) style formulation. We consider the following two problems, assuming we can sample parameters identically and independently distributed from $\mathbb{P}$.
\begin{problem}
\label{problem:PACbound}
    For any $C \subset S$, any $\mathcal{C}, \mathcal{S} \subset 2^S$, any $N > 0$, and any $\beta \in (0,1)$, compute $\eta_N(C, \beta)$ and $\zeta_N(\mathcal{C}, \mathcal{S}, \beta)$  that satisfy the following conditions:
    \begin{align}
    \label{problem:eta}
        &\mathbb{P}^N( F(\mathcal{M}_\mathbb{P}, C) \geq \eta_N(C, \beta) ) \geq \beta, \\
    \label{problem:zeta}
        &\mathbb{P}^N( R(\mathcal{M}_\mathbb{P}, \mathcal{C}, \mathcal{S}) \geq \zeta_N(\mathcal{C}, \mathcal{S}, \beta) ) \geq \beta,
    \end{align}
    where $\mathbb{P}^N$ is the $N$ product measure of $\mathbb{P}$.
\end{problem}
\begin{problem}
\label{problem:PACoptimal_diff}
    For any $N > 0$ and any $\delta, \beta \in (0,1)$, compute $\mathcal{C}^* \subset 2^S$ that satisfies the following two conditions:
    \begin{description}
        \item[(PC1)] $\mathcal{C}^* \in \argmax_{\mathcal{C} \subseteq \mathcal{S}_{N, \delta, \beta}} \zeta_N(\mathcal{C}, \mathcal{S}_{N, \delta, \beta}, \beta)$, 
        \item[(PC2)] $\forall \mathcal{C}' \subset \mathcal{C}^*$, \\ $\mathbb{P}^N(  R(\mathcal{M}_\mathbb{P}, \mathcal{C}^*, \mathcal{S}_{N, \delta, \beta}) > R(\mathcal{M}_\mathbb{P}, \mathcal{C}', \mathcal{S}_{N, \delta, \beta}) ) > \beta$,
    \end{description}
    where both $\eta_N$ and $\zeta_N$ are the solutions to Problem \ref{problem:PACbound}, and $\mathcal{S}_{N, \delta, \beta}$ is defined as \begin{align}
    \label{S_delta_beta}
        \mathcal{S}_{N, \delta, \beta} = \{ C \subset 2^S \;|\; \eta_N(C, \beta) > \delta \}.
    \end{align}
\end{problem}
In Problem \ref{problem:PACbound}, we compute probabilistically guaranteed lower bounds for the SPR-cause probability and recall-optimal probability. Problem \ref{problem:PACoptimal_diff} relaxes Problem \ref{problem:exact} based on the PAC bounds.

\section{Probability-Raising Causal Identification on Uncertain Parametric MDPs}

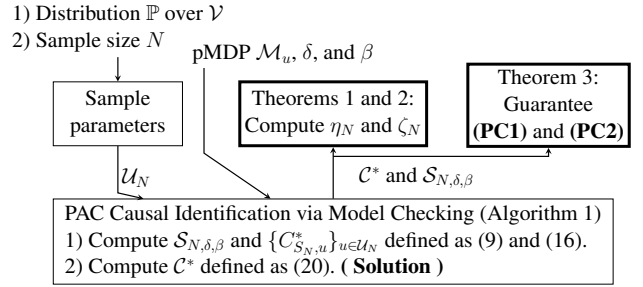
\begin{figure}[htbp]
\centering
\resizebox{\linewidth}{!}{\begin{tikzpicture}[node distance=1cm]
    % ノードの定義
    \node[draw, rectangle, minimum width=13cm, minimum height=2cm, align=left, font=\Large] (A) at (0, 0) {
        PAC Causal Identification via Model Checking (Algorithm \ref{Alg_PG_SPREst}) \\
        1) Compute $\mathcal{S}_{N, \delta, \beta}$ and $\{ C^*_{S_N, u} \}_{u \in \mathcal{U}_N}$ defined as (\ref{S_delta_beta}) and (\ref{canonical_cause}).\\
        2) Compute $\mathcal{C}^*$ defined as (\ref{C_star}). \textbf{( Solution )}
    };

    \node[draw=none, align=left, font=\Large] (E) at (-5, 5) {
        1) Distribution $\mathbb{P}$ over $\mathcal{V}$\\
        2) Sample size $N$
    };

    \node[draw, rectangle, minimum width=3cm, minimum height=1.5cm, align=center, font=\Large] (B) at (-5, 3) {Sample\\ parameters};
    \node[draw, rectangle, line width=0.7mm, minimum width=4cm, minimum height=1.5cm, align=center, font=\Large] (C) at (0, 3) {Theorems 1 and 2:\\ Compute $\eta_N$ and $\zeta_N$};
    \node[draw, rectangle, line width=0.7mm, minimum width=3cm, minimum height=1.5cm, align=center, font=\Large] (D) at (5.0, 3.2) {Theorem 3:\\ Guarantee\\ \textbf{(PC1)} and \textbf{(PC2)}};

    \node[draw=none] (Empty) at (-3, 4.2) {};
    \node[draw=none, align=left, font=\Large] (F) at (-1.1, 4.4) {
        pMDP $\mathcal{M}_u$, $\delta$, and $\beta$
    };

    \node[draw=none, align=left, font=\Large] (G) at (2.0, 1.5) {
        $\mathcal{C}^*$ and $\mathcal{S}_{N, \delta, \beta}$
    };

    \node[draw=none, align=left, font=\Large] (H) at (-4.5, 1.5) {
        $\mathcal{U}_N$
    };

    % 矢印の定義
    \draw[->, >=stealth, thick] (E) -- (B);
    \draw[->, >=stealth, thick] (B.south) -- ++(0, -1.1) -- (A);
    \draw[->, >=stealth, thick] (A) -- ++(0, 2) coordinate (mid) -- (C);
    \draw[->, >=stealth, thick] (mid) -| (D);

    \draw[->, >=stealth, thick] (Empty.south) -- ++(0,-2) -- (A);

\end{tikzpicture}}
\caption{Overview of our proposed method. We first sample the parameters of the upMDP. Then, we compute the solution to Problem \ref{problem:PACoptimal_diff} from the sampled MDPs based on model checking. We provide the probabilistic guarantees for the obtained result using Theorems \ref{thm_certif}, \ref{thm_optimality}, and \ref{thm_solution_certif}.}
\label{fig:overview}
\end{figure}

In this section, we describe how we obtain the solutions to Problems \ref{problem:PACbound} and \ref{problem:PACoptimal_diff}.
In Section \ref{subsection:ProbGuarantee}, we provide the solutions $\eta_N$ and $\zeta_N$ to Problem \ref{problem:PACbound} using parameter sampling.
In Section \ref{subsection:Estimation}, we construct a solution to Problem \ref{problem:PACoptimal_diff} based on a set covering for the sampled parameters.
In Section \ref{subsection:Algorithm}, we provide a model checking-based algorithm to compute the solution to Problem \ref{problem:PACoptimal_diff}.
We show the overview of our approach in Fig.\ \ref{fig:overview}. 
In the following, we denote by $\mathcal{U}_N = (u_i)_{i=1}^N$ the sequence of $N$ parameter vectors sampled from the parameter space $\mathcal{V}$ of the upMDP $\mathcal{M}_\mathbb{P}$.

\subsection{PAC Guarantees for Probability-Raising Causes on Uncertain Parametric MDPs}
\label{subsection:ProbGuarantee}

In this section, we provide a solution to Problem \ref{problem:PACbound}.
We first give a lower bound of the SPR-cause probability for any subset of states $C \subseteq S \setminus E$ as a solution to (\ref{problem:eta}).
For any $u \in \mathcal{V}$ and any $S' \subseteq S$, we define $\mathscr{C}_{S',u} \subset S$ as 
\begin{align}
    \label{all_cause_states}
    \mathscr{C}_{S',u} = \{ c \in S' \;|\; (\mathcal{M}_u, \{c\}) \models \textbf{(S)} \},
\end{align}
where \textbf{(S)} is defined in Def. \ref{def:SPR_cause}.
For any $C \subseteq S$, we define $n(C, \mathcal{U}_N)$ as the number of samples $u \in \mathcal{U}_N$ that satisfy the following condition:
\begin{align}
\label{n_count}
    C \subseteq \mathscr{C}_{S,u}, (\mathcal{M}_u, C) \models \textbf{(M)}.
\end{align}
Note that $n(C, \mathcal{U}_N)$ is regarded as the number of times that each $c \in C$ is an SPR cause on $\mathcal{M}_u$ and $C$ satisfies the condition \textbf{(M)} of SPR cause in Def. \ref{def:SPR_cause}.
Then, we define the function $\eta_N : 2^S \times (0,1) \to \mathbb{R}$ as
\begin{align}
\label{eta}
    \eta_N(C, \beta) = t^*(N - n(C, \mathcal{U}_N), \beta),
\end{align}
where the function $t^* : [N] \times [0,1] \to \mathbb{R}$ is defined such that, for any $\beta$, 
\begin{align}
\label{t_star_maxmin}
    t^*(0, \beta) = (1 - \beta)^{1/N}, \ t^*(N, \beta)=0,
\end{align}and, for any $k = 1, \ldots, N-1$, $t^*(k, \beta)$ is the solution of
\begin{align}
\label{t_star}
    \frac{1 - \beta}{N} = \sum_{i=0}^k \binom{N}{i} (1 - t)^i t^{N-i}.
\end{align}

\begin{theorem}
\label{thm_certif}
    For any upMDP $\mathcal{M}_\mathbb{P}$, any $N>0$, any $C \subseteq S \setminus E$, and any $\beta \in (0, 1]$, we have
    \begin{align}
        \mathbb{P}^N \{ F(\mathcal{M}_\mathbb{P}, C) \geq \eta_N(C, \beta) \} \geq \beta,
    \end{align}
    where $\eta_N$ is defined as (\ref{eta}).
\end{theorem}
By Theorem \ref{thm_certif}, $\eta_N$ defined as (\ref{eta}) is a solution to (\ref{problem:eta}).

We provide the lower bound of the recall-optimal probability of any collection of subsets of states as a solution to (\ref{problem:zeta}).
For any $u \in \mathcal{V}$ and any $S' \subset S$, we define a canonical SPR cause $C^*_{S',u}$ for $S'$ as
\begin{align}
\label{canonical_cause}
    C^*_{S',u} \! =\! \{ c \in \mathscr{C}_{S',u} | \exists \pi \mbox{ s.t. } \mathrm{Pr}_{\mathcal{M}_u}^\pi(\neg \mathscr{C}_{S',u} \mathrm{U} c) > 0 \},
\end{align}
Intuitively, $C^*_{S',u}$ is interpreted as the ``front" of all singleton SPR causes over $S'$.
For any $\mathcal{C}, \mathcal{S} \subseteq 2^{S}$, we define $m(\mathcal{C}, \mathcal{S}, \mathcal{U}_N )$ as the number of samples $u \in \mathcal{U}_N$ such that there exists $C \in \mathcal{C}$ that satisfies the following three conditions: (i) $C$ covers all path reaching $E$ that passes through $C^*_{S',u}$ for the union $S'$ of $\mathcal{S}$, (ii) $C$ is a subset of $\mathscr{C}_{S',u}$, and (iii) $C$ satisfies \textbf{(M)} in Def. \ref{def:SPR_cause}, that is,
\begin{align}
\label{m_count}
    \exists C \in \mathcal{C}& \mbox{ s.t. } C \subseteq \mathscr{C}_{S',u}, (\mathcal{M}_u, C) \models \textbf{(M)}, \nonumber \\
    & \forall \rho \in \mathrm{FinPath}, \rho \models \Diamond C^*_{S',u} \land \Diamond E \implies \Diamond C,
\end{align}
where $\rho \models \varphi$ denotes that the path $\rho$ on an MDP $\mathcal{M}_u$ satisfies the formula $\varphi$ and $S' = \cup_{C' \in \mathcal{S}} C'$.
We define the function $\zeta_N : 2^{2^S} \times 2^{2^S} \times (0,1) \to \mathbb{R}$ as
\begin{align}
\label{zeta}
    \zeta_N(\mathcal{C}, \mathcal{S}, \beta) = t^*(N - m(\mathcal{C}, \mathcal{S}, \mathcal{U}_N), \beta),
\end{align}
where $t^*$ is defined as (\ref{t_star}).

\begin{theorem}
\label{thm_optimality}
    For any upMDP $\mathcal{M}_\mathbb{P}$, any $N>0$, any $\mathcal{C}, \mathcal{S} \subseteq 2^{S}$, and any $\beta \in (0,1]$, we have
        \begin{align}
        \label{PR-recall_optimality_bound}
            &\mathbb{P}^N \{ R(\mathcal{M}_\mathbb{P}, \mathcal{C}, \mathcal{S}) \geq \zeta_N(\mathcal{C}, \mathcal{S}, \beta) \} \geq \beta,
        \end{align}
        where $\zeta_N$ is defined as (\ref{zeta}).
\end{theorem}
By Theorem \ref{thm_optimality}, $\zeta_N$ defined by (\ref{zeta}) is a solution to (\ref{problem:zeta}).

\begin{remark}[On Lower Bounds in Theorems \ref{thm_certif} and \ref{thm_optimality}]
    The right-hand side of (\ref{t_star}) is the cumulative distribution function of the binary distribution for $k$ with the success probability $1 - t$. Clearly, the maximum and minimum values of $t^*$ defined by (\ref{t_star_maxmin}) and (\ref{t_star}) are $(1-\beta)^{1/N}$ and $0$.
    Thus, the values of $\eta_N(C, \beta)$ and $\zeta_N(\mathcal{C}, \mathcal{S}, \beta)$ defined by (\ref{eta}) and (\ref{zeta}), respectively, are strictly decreasing for $\beta$, and they converge to $0$ as $\beta$ goes to $1$. Moreover, $t^*$ is strictly decreasing for $k$. Hence, $\eta_N(C, \beta)$ and $\zeta_N(\mathcal{C}, \mathcal{S}, \beta)$ are strictly increasing for $n(C, \mathcal{U}_N)$ and $m(\mathcal{C}, \mathcal{S}, \mathcal{U}_N)$, respectively, and they converge to $(1 - \beta)^{1/N}$ when $n(C, \mathcal{U}_N)$ and $m(\mathcal{C}, \mathcal{S}, \mathcal{U}_N)$ go to $N$. Please see Section 5.4 in \citep{badings2022scenario} for a more detailed discussion about the behavior of $t^*$.
\end{remark}

\begin{remark}[Proof idea for Theorems \ref{thm_certif} and \ref{thm_optimality}]
Our probabilistic lower bounds for SPR-cause and recall-optimal probabilities are based on optimization formulations known as \textit{scenario optimization programs} (SOPs) \citep{campi2008exact, campi2011sampling}. We consider two SOPs whose
constraints require that a given set of states $C \subseteq S$ is an SPR cause, and a given collection of sets of states $\mathcal{C} \subseteq 2^S$ contains a recall-optimal SPR cause, respectively, in some sampled MDPs. Please see Appendices \ref{appendix:scenario} and \ref{Appendix_thm12} for details.
\end{remark}

\subsection{PAC-Guaranteed Recall-Optimal Probability-Raising Causes on Uncertain Parametric MDPs}
\label{subsection:Estimation}

In this section, we provide the solution to Problem \ref{problem:PACoptimal_diff}.
For any $\beta, \delta \in (0,1)$ and any $N>0$, we define the collection $\mathcal{C}^*$ of subsets of states as
\begin{align}
\label{C_star}
    \mathcal{C}^* = \{ C^*_{S_N, u_i} \;|\; i \in \mathcal{I} \},
\end{align}
where $\mathcal{I}$ is a set of indices for sampled parameter vectors defined such that
\begin{align}
    \bigcup_{i \in \mathcal{I}} \mathcal{U}^i_N = \mathcal{U}_N \land \bigcup_{i \in \mathcal{I}'}\mathcal{U}^i_N \neq \mathcal{U}_N, \forall \mathcal{I}' \subset \mathcal{I},
\end{align}
$\mathcal{U}^i_N$ is defined as
\begin{align}
\label{recall_optimal_samples}
    \mathcal{U}^i_N = & \{ u_j \in \mathcal{U}_N \;|\; C^*_{S_N, u_i} \subseteq \mathscr{C}_{u_j} \nonumber \\
    & \land \forall \rho, \rho \models \Diamond C^*_{S_N, u_j} \land \Diamond E \implies \Diamond C^*_{S_N, u_i} \}
\end{align}
and $C^*_{S_N, u_i}$ is defined by (\ref{canonical_cause}) with $S_N = \cup_{C' \in \mathcal{S}_{N, \delta, \beta}} C'$.

\begin{theorem}
\label{thm_solution_certif}
    For any upMDP $\mathcal{M}_\mathbb{P}$, any $\delta, \beta \in (0,1)$, and any $N>0$, $\mathcal{C}^*$ defined by (\ref{C_star}) is a solution to Problem \ref{problem:PACoptimal_diff}.
\end{theorem}

In the following Example \ref{example1_SPR}, we show how the condition in (\ref{recall_optimal_samples}) eliminates redundant samples for the recall-optimal probability from all sampled parameters and how to construct the exhaustive solution $\mathcal{C}^*$ defined by (\ref{C_star}).
\begin{example}
    \label{example1_SPR}
    We consider solving Problem \ref{problem:PACoptimal_diff} for the upMDP in Example \ref{example3}. Let $\delta=0$ and $\beta=0.99$. We sampled $N=1000$ parameters. We obtained a solution $\mathcal{C}^* = \{ \{s_3\} \}$ with $\mathcal{S}_{N, \delta, \beta}$ $ =\{ \{s_1\}, \{s_2\}, \{s_3\}, \{s_1, s_3\}, \{s_2, s_3\} \}$, where $\{s_3\}$, $\{s_1, s_3\}$, and $\{s_2, s_3\}$ are the canonical SPR causes on the sampled $7$, $2$, and $991$ MDPs with $p = q$, $p > q$, and $p < q$, respectively. $\{s_3\}$ is a subset of $\{s_1, s_3\}$ and $\{s_2, s_3\}$, and all paths passing through $\{s_1, s_3\}$ or $\{s_2, s_3\}$ go through $\{s_3\}$. Thus, $\mathcal{U}^i_N$ defined by (\ref{recall_optimal_samples}) that corresponds to $\{s_3\}$ is equal to $\mathcal{U}_N$. According to Theorem \ref{thm_certif}, the subsets of states $\{s_1\}$, $\{s_2\}$, $\{s_3\}$, $\{s_1, s_3\}$, and $\{s_2, s_3 \}$ provide the lower bounds of SPR-cause probabilities $4.3 \times 10^{-6}$, $0.98$, $0.995$, $4.3 \times 10^{-6}$, and $0.98$ with the confidence probability $\beta=0.99$, respectively, and they exceed $\delta$. For \textbf{(PC1)}, the lower bound $\zeta_N(\mathcal{C}^*, \mathcal{S}_{N, \delta, \beta}, \beta)$ of recall-optimal probability is the maximum value $(1 - \beta)^{1/N} = 0.995$. 
    For \textbf{(PC2)}, the subset $\mathcal{C} = \emptyset \subset \mathcal{C}^*$ clearly does not contain any SPR cause for all parameters. So, $R(\mathcal{M}_\mathbb{P}, \mathcal{C}^*, S_{N, \delta, \beta}) > R(\mathcal{M}_\mathbb{P}, \mathcal{C}, S_{N, \delta, \beta})$ holds.
\end{example}
\begin{remark}[Interpretation for (\ref{C_star})]
    Intuitively, $\mathcal{U}^i_N$ defined by (\ref{recall_optimal_samples}) collects all sampled parameters $u_j$ such that $C^*_{S_N, u_i}$ for $i$-th sample $u_i$ is recall optimal on $\mathcal{M}_{u_j}$ (See Lemma \ref{lemma_recall_optimal_necessary_sufficient} in Appendix \ref{Appendix_thm3}). So, $\mathcal{C}^*$ is constructed to contain some recall-optimal SPR causes for each sampled MDP.
\end{remark}
An easier example is shown in Appendix \ref{appendix:other_example}. Further, additional discussion about the interpretability of the resulting SPR causes is provided in Appendix \ref{appendix:interpretability}.

%% Algorithm
\subsection{Causal Identification Algorithm}
\label{subsection:Algorithm}
We provide an algorithm to obtain the solution to Problem \ref{problem:PACoptimal_diff}. The overall procedure is shown in Algorithm \ref{Alg_PG_SPREst}.
In Lines from 1 to 6, we compute the canonical SPR cause $C^*_{S_N, u}$ defined as (\ref{canonical_cause}) for each sampled parameter vector, where $S_N = \cup_{C' \in \mathcal{S}_{N,\delta,\beta}}C'$, for each $\mathcal{M}_u$ with $u \in \mathcal{U}_N$ by applying Algorithm \ref{Alg_SinCheck} to all states in $S \setminus E$.
Then, in Line 8, we obtain $\mathcal{C}^*_{\mathcal{U}_N, \delta, \beta}$ defined by (\ref{C_star}) as the solution to Problem \ref{problem:PACoptimal_diff}. We compute the lower bounds of SPR-cause probabilities and recall-optimal probability through (\ref{eta}) and (\ref{zeta}). 
Note that we can check the conditions of (\ref{canonical_cause}), (\ref{m_count}), and (\ref{C_star}) using a standard graph algorithm for path finding over $\mathcal{M}_u$.
Algorithm \ref{Alg_SinCheck} is identical to Algorithm 2 in \citep{baier2022probability} and checks whether a given state $\{c \}$ is an SPR cause on a given MDP $\mathcal{M}_u$, that is, $(\mathcal{M}_u, \{c\}) \models \text{(S)}$. The computation of $\tau_u$ is sound and complete, i.e.,  $\tau_u(c) = 1$ if and only if $\{ c \}$ is an SPR cause on $\mathcal{M}_u$, by Lemma 4 in \citep{baier2022probability}. In Lines 1 and 2, we compute the minimum probability $w_c$ of reaching $E$ from $c$ on $\mathcal{M}_u$ and the maximum probability $q_s$ of reaching $E$ from any $s \in S$ on a modified MDP $\mathcal{M}_u^{[c]}$. The computation is conducted by existing model checking algorithms such as value iteration. 
The modified MDP $\mathcal{M}_u^{[c]}$ is defined by removing all actions at $c$ and adding a terminal state $\mathrm{noeff}$ that is not in $E$ and an action $\gamma$ that is enabled at $c$ with the transition probabilities $P_u(c, \gamma, \mathrm{eff}) = \mathrm{Pr}^\mathrm{min}_{\mathcal{M}_u,c}(\Diamond E)$ and $P_u(c, \gamma, \mathrm{noeff}) = 1 -\mathrm{Pr}^\mathrm{min}_{\mathcal{M}_u,c}(\Diamond E)$, where $\mathrm{eff}$ is a fixed state in $E$ and $\mathrm{Pr}^\mathrm{min}_{\mathcal{M}_u,c}(\Diamond E) = \min_{\pi} \mathrm{Pr}^\pi_{\mathcal{M}_u,c}(\Diamond E)$. $\tau_u(c)$ is calculated depending on the following 3 cases. When $w_c - q_{s_0} > 0$ and $w_c - q_{s_0} < 0$, we set $\tau_u(c) = 1$ and $-1$, respectively. In the corner case $w_c - q_{s_0} = 0$, we consider the set of actions $\mathcal{A}_{res}(s)\{ a \in A | q_s\! =\! \sum_{s' \in S^{[c]}} P_u(s,a,s') q_{s'} \}$ for any state $s$, where $S^{[c]}$ is the set of states in $\mathcal{M}^{[c]}_u$. We construct a sub-MDP $\mathcal{M}^{\mathrm{max}, [c]}_u$ of $\mathcal{M}^{[c]}_u$ induced from $\mathcal{A}_{res}$. In Lines from 8 to 12, if $c$ is reachable from $s_0$ in the sub-MDP $\mathcal{M}^{\mathrm{max},[c]}_u$, then we set $\tau_u(c) = -1$, and otherwise $\tau_u(c) = 1$.
The complexity of Algorithm \ref{Alg_PG_SPREst} is polynomial for $|S|$ and $N$ in practice. Please see Appendix \ref{appendix:complexity} for details.

\begin{algorithm}[htbp]
\caption{Probabilistically Guaranteed Recall-Optimal SPR Cause Identification for upMDP}
\begin{algorithmic}[1] 
\label{Alg_PG_SPREst}
\REQUIRE upMDP $\mathcal{M}_\mathbb{P}$, $N > 0$, $\delta \geq 0$, and $\beta \in [0,1]$.
\ENSURE $\mathcal{C}^*$, $\eta_N$, and $\zeta_N$.
\STATE Sample $N$ parameters $\mathcal{U}_N$ from $\mathcal{V}$ according to $\mathbb{P}$.
\FOR{each $u \in \mathcal{U}_N$}
\FOR{each $c \in S \setminus E$}
\STATE $\tau_u(c) \gets$ Algorithm\ref{Alg_SinCheck}($\mathcal{M}_u, c)$.
\ENDFOR
\STATE Compute $\mathscr{C}_{S_N, u}$ as $\mathscr{C}_{S_N, u} \! =\! \{ c \in S | \tau_u(c) = 1 \}$ and $C^*_{S_N,u}$ defined by (\ref{canonical_cause}), where $S_N = \cup_{C' \in \mathcal{S}_{N,\delta,\beta}} C'$.
\ENDFOR
\STATE Initialize $\mathcal{I} = \emptyset$.
\WHILE{$\bigcup_{i \in \mathcal{I}} \mathcal{U}^i_N \neq \mathcal{U}_N$.}
\STATE $\mathcal{I} \gets \mathcal{I} \cup\{i\}$ such that $\mathcal{U}^i_N \setminus \bigcup_{j \in \mathcal{I}} \mathcal{U}^j_N \neq \emptyset$.
\ENDWHILE
\STATE Compute $\mathcal{C}^*$ from $\mathcal{I}$.
\STATE Compute $\eta_N$ and $\zeta_N$ for $\mathcal{C}^*$ and $\mathcal{S}_{N, \delta, \beta}$ using Theorems \ref{thm_certif} and \ref{thm_optimality}.
\end{algorithmic}
\end{algorithm}

\begin{algorithm}[htbp]
\caption{SPR Cause Checking for Single State}
\begin{algorithmic}[1] 
\label{Alg_SinCheck}
\REQUIRE MDP $\mathcal{M}_u$, $c \in S$.
\ENSURE $\tau_u(c) \in \{-1, 1\}$.
\STATE Compute $w_c = \min_\pi \mathrm{Pr}^{\pi}_{\mathcal{M}_u, c}(\Diamond E)$.
\STATE Compute $q_{s} = \max_\pi \mathrm{Pr}^{\pi}_{\mathcal{M}^{[c]}_u, s}(\Diamond E)$ for each $s \in S$.
\IF{ $w_c - q_{s_0} > 0$ }
\STATE $\tau_u(c) = 1$
\ELSIF{$w_c - q_{s_0} < 0$}
\STATE $\tau_u(c) = -1$.
\ELSIF{$w_c - q_{s_0} = 0$}
\IF{ c is reachable from $s_0$ in $\mathcal{M}^{\mathrm{max},[c]}_u$}
\STATE $\tau_u(c) = -1$.
\ELSE
\STATE $\tau_u(c) = 1$.
\ENDIF
\ENDIF
\end{algorithmic}
\end{algorithm}

\section{Example}
\subsection{Setup}
We demonstrate that the proposed method exhaustively and nonredundantly finds SPR-causes on a $10 \times 10$ grid world with $3$ obstacles, where the robot behaves as depicted in Figs. \ref{ex:env1} and \ref{ex:env2}. The undesired behavior is reaching the terminal red cells. The blue arrows represent the enabled actions for each set of cells surrounded by the green dashed line. 
The actions the robot can take are $\mathrm{Up}, \mathrm{Right}, \mathrm{Down}, \mathrm{Left}, \mbox{ and } \mathrm{Stay}$. 
In Fig. \ref{ex:env1}, the cells $\{(4, i) \}_{i=6,7,8}$ is the one-way aisle where the robot is allowed to choose to go up. On the other hand, in Fig. \ref{ex:env2}, the robot can choose to move up and down the aisle. 
The robot transits in the intended direction with probability $p_0$ and moves perpendicularly with probability $(1 - p_0)/2$. The robot always stays in the same cell when it chooses $\mathrm{Stay}$ or is about to collide with the obstacles or go outside. The robot enters the red cell with probability $p_1$ (resp., $p_2$) regardless of the selected action when it is in the cell $(9, 5)$ (resp., $(8, 9)$) and moves in the intended way with the remaining probability. The robot moves carefully in the intended direction with probability 1 in the cells adjacent to the red cells, other than $(9, 5)$ and $(8, 9)$. 
We set the probability distributions followed by $p_0$, $p_1$, and $p_2$ as the uniform distributions over the intervals $[0.85, 0.9]$, $[0.45, 0.6]$, and $[0.5, 0.7]$. We sampled $N=100$ parameter vectors from the parameter space. We conducted the experiments $10$ times for each environment with $\beta= \beta_0 (=0.9), \beta_1(=0.99)$ and $\delta=\delta_0(=0.001), \delta_1(=0.1)$. We implemented Algorithm \ref{Alg_SinCheck} based on value iteration to compute $w_c$ and $q_s$ and to check some path conditions. We computed ${\eta}_N$ and ${\zeta}_N$ defined in (\ref{eta}) and (\ref{zeta}), respectively, using bi-section search.

SPR causes are cells whose traversal increases the probability of reaching red cells.
All true SPR causes when $p_1 < p_2$ and $p_1 > p_2$ are shown as the yellow and orange/green cells, respectively. This is because crossing the yellow cells leads to entering the red cells by going through the top side when $p_1 < p_2$. Likewise, passing through the orange/green cells causes entry into the red cells from below when $p_1 > p_2$.
Our method identifies and consolidates them through model checking for sampled MDPs and set covering.
A detailed explanation for SPR causes in Example is described in Appendix \ref{appendix:explain_example_sec}.

We compared our proposed method with two na\"ive approaches: $\mathrm{NA1}$ and $\mathrm{NA2}$. $\mathrm{NA1}$ computes the canonical SPR cause $C^*_{S,u}$ defined by (\ref{canonical_cause}) based on Algorithm 2 in \citep{baier2022probability} for the mean parameters. $\mathrm{NA2}$ collects the canonical SPR causes $C^*_{S,u}$ for 8 parameter vectors corresponding to the vertices of the parameter region.

\begin{figure}[htbp]
  	\centering
  	\subfigure[]{
  		\includegraphics[width=0.45\linewidth]{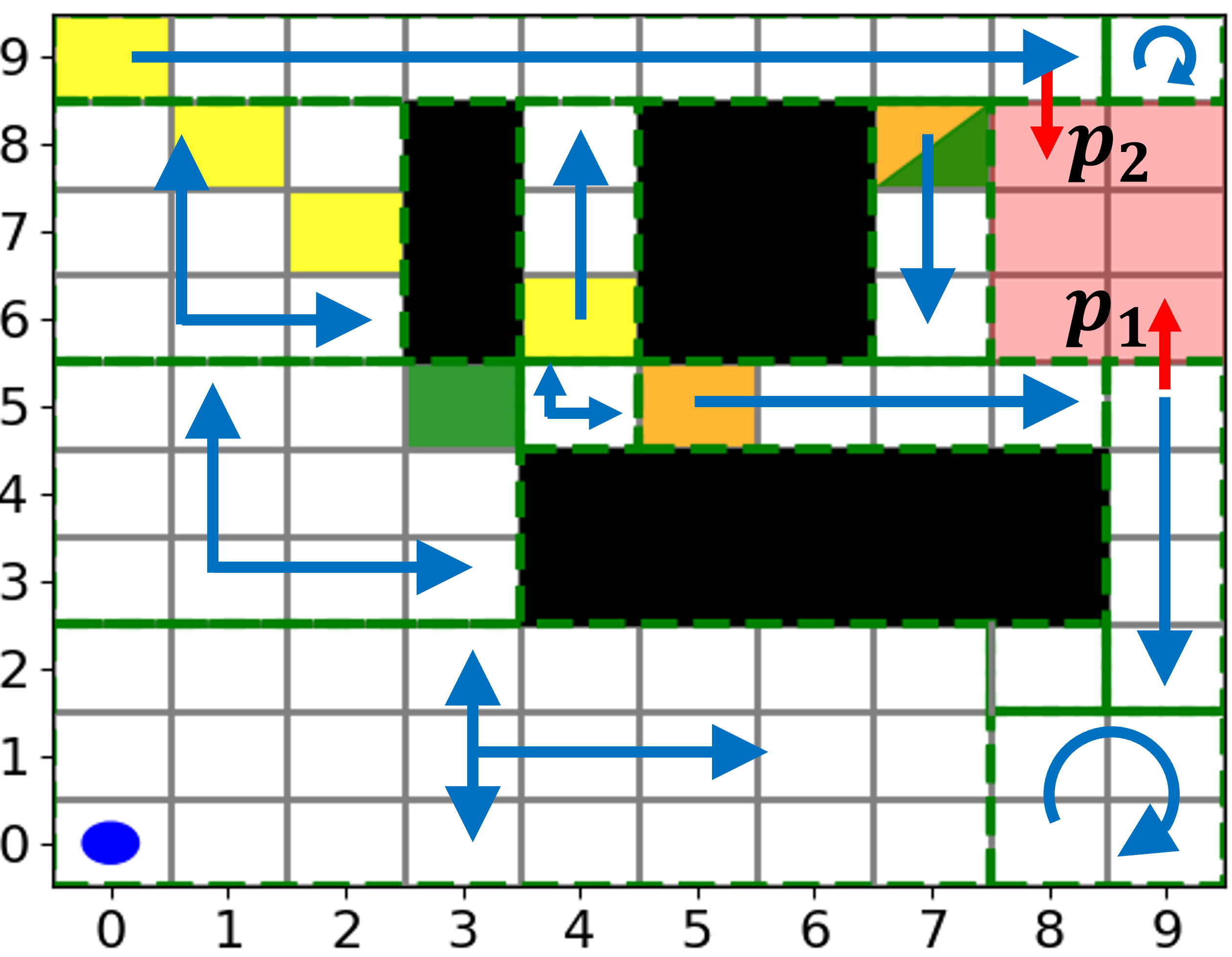}
            \label{ex:env1}
  	}
  	\subfigure[]{
  		\includegraphics[width=0.45\linewidth]{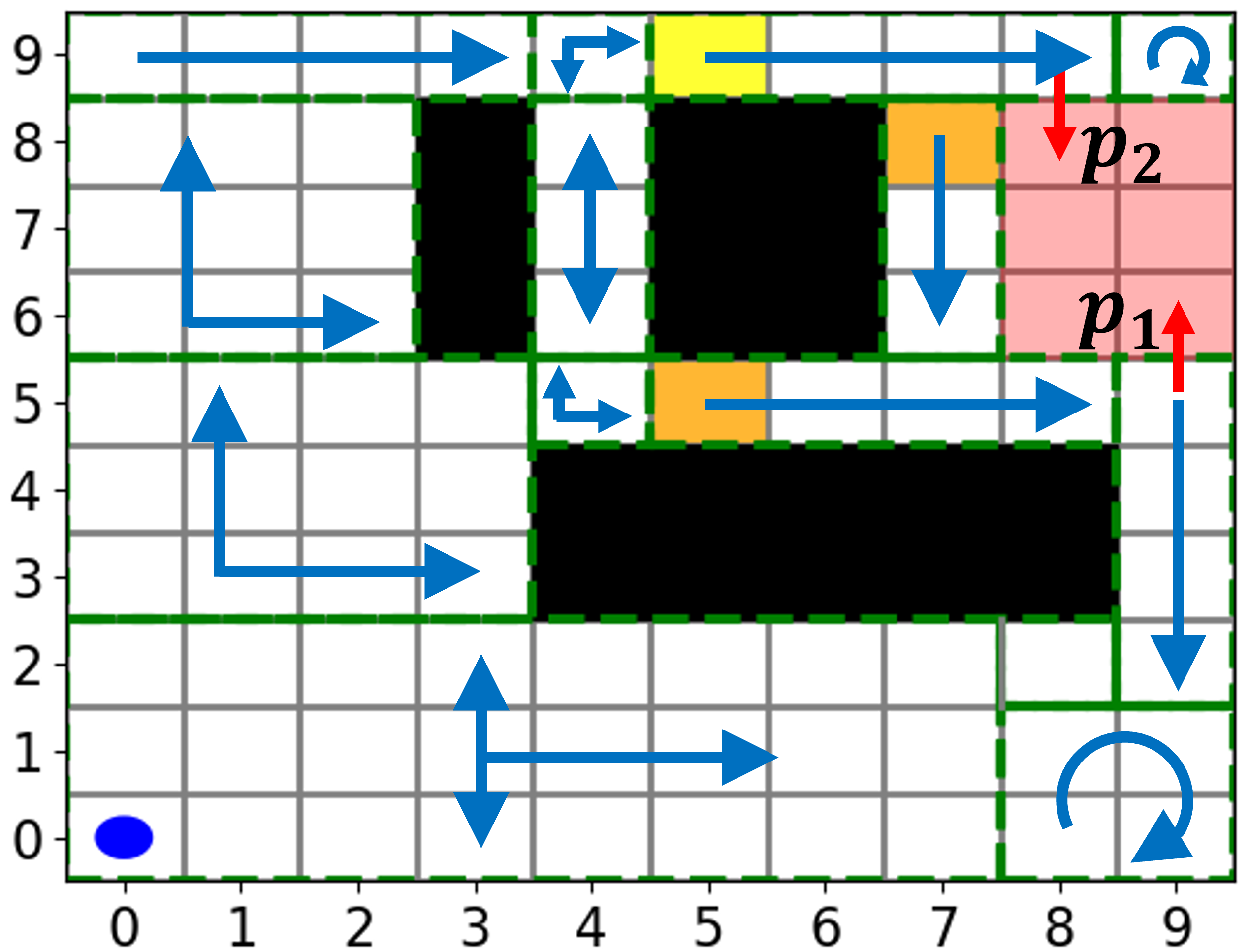}
  		\label{ex:env2}
  	}
  	\caption{Two grid-worlds with three obstacles (black rectangles). The robot starts from $(0,0)$. It enters the terminal red cells by the red arrow transitions with the indicated probabilities $p_1$ and $p_2$. The blue arrows denote the actions enabled in the regions. The yellow and orange/green cells denote the SPR cause when $p_1 < p_2$ and $p_1 > p_2$, respectively, obtained from our method with $\delta=0.001$. The difference between the two environments is the enabled actions in the top middle cells.}
   \vspace{0mm}
\end{figure}

\subsection{Results}
\begin{figure*}[htbp]
  	\centering
  	\subfigure[]{
  		\includegraphics[width=0.305\linewidth]{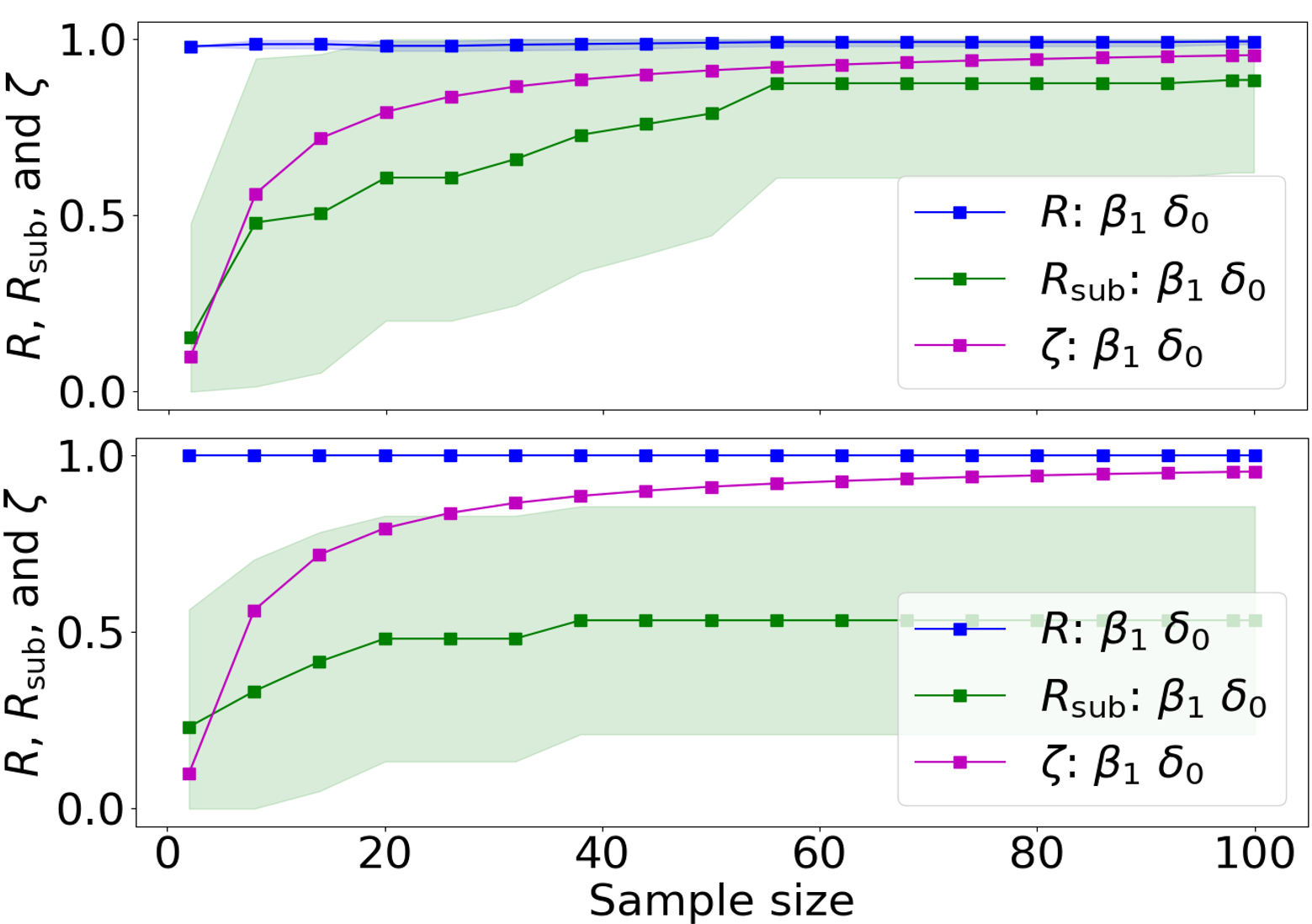}
            \label{ex:result_R_zeta}
  	}
  	\subfigure[]{
  		\includegraphics[width=0.305\linewidth]{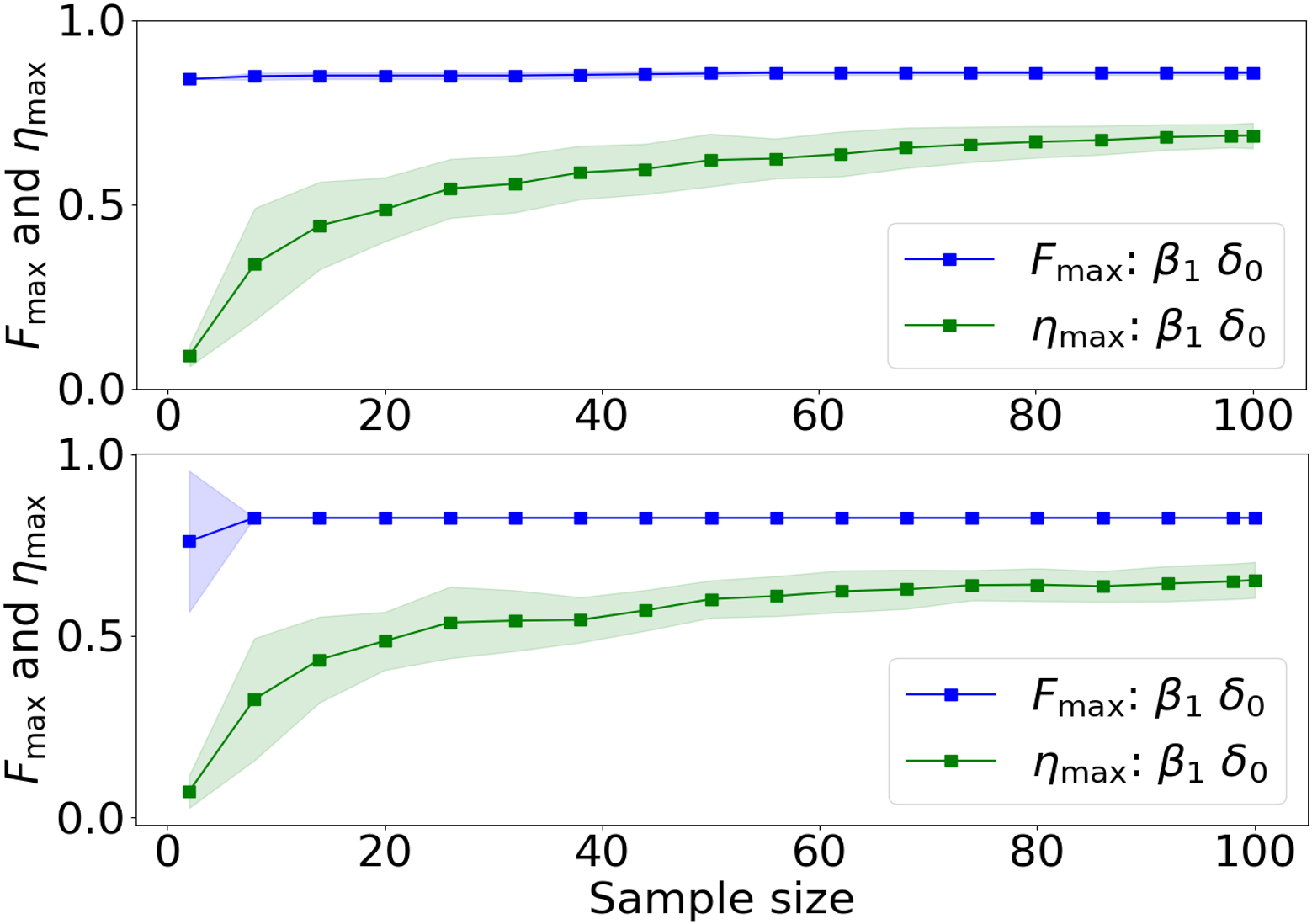}
  		\label{ex:results_F_eta_max}
  	}
        \subfigure[]{
  		\includegraphics[width=0.305\linewidth]{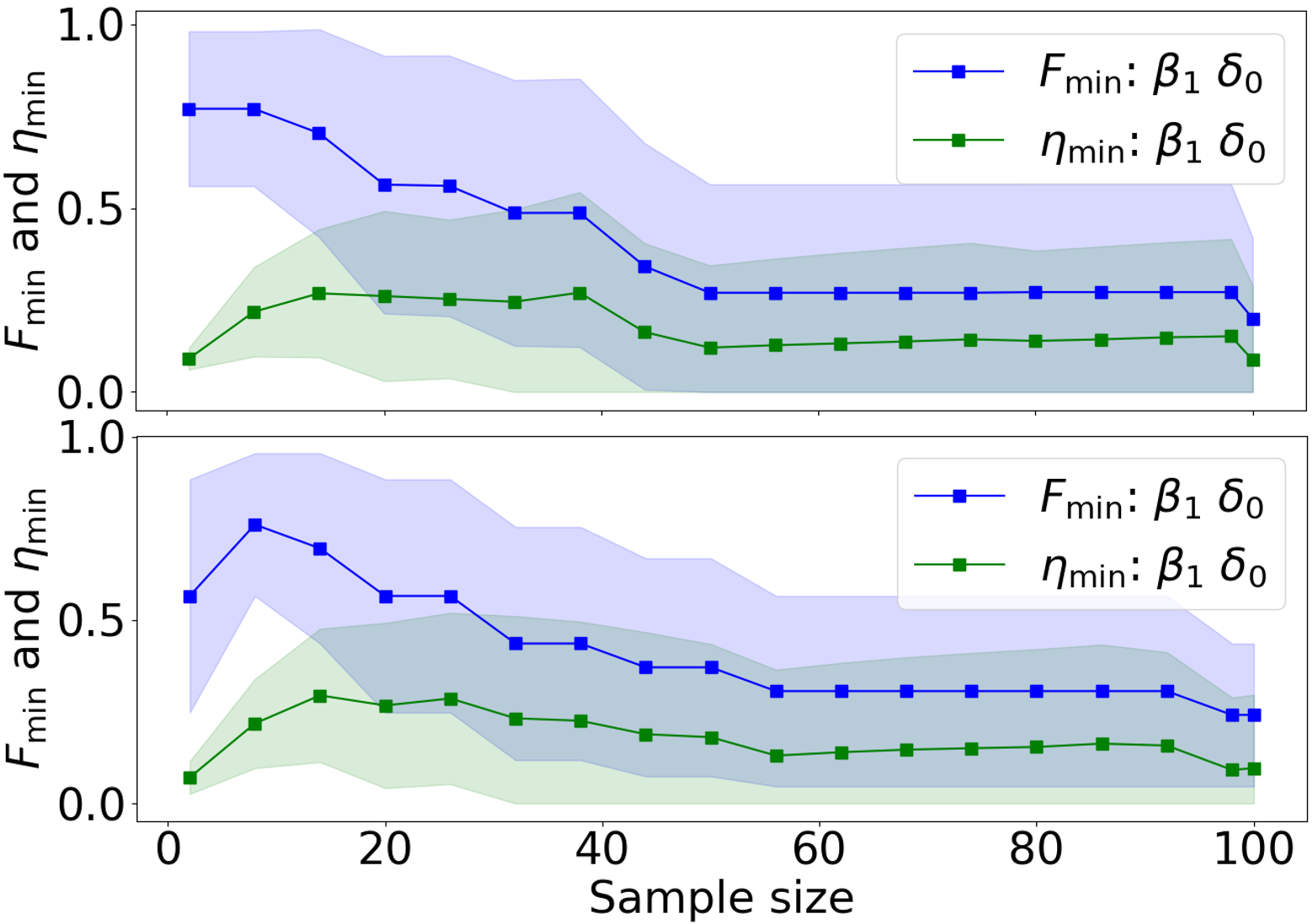}
  		\label{ex:results_F_eta_min}
  	}
  	\caption{Results for $\delta_0$ and $\beta_1$. (a): The means and standard deviations of the recall-optimal probabilities $R$ and their lower bounds $\zeta$ for the solutions obtained from the proposed method against each sample size (blue and magenta). We further show the maximum recall-optimal probabilities $R_\mathrm{sub}$ for all proper subsets of the solutions (green). (b) (resp., (c)): The means and standard deviations of the maximum (resp., minimum) SPR-cause probabilities and their lower bounds over all potential SPR causes in the solutions obtained from the proposed method against each sample size. The results of the upper and the lower figures in (a)-(c) correspond to two environments depicted in Fig\ \ref{ex:env1} and \ref{ex:env2}, respectively.}
   \vspace{0mm}
\end{figure*}

%% 得られる解に対して1000サンプルで近似したRとzetaを比較．Fとetaを比較
% Problem2と3を解けていることを示す．
% NA1とNA2は解の不足や定量化の不足でProblem2,3を解けていないことを示す．
We show that the proposed method finds solutions to Problem \ref{problem:PACoptimal_diff} well for various sample sizes.
We show the results for $\delta_0$ and $\beta_1$ (See Appendix \ref{appendix:detailed_results} for more results).
Fig.\ \ref{ex:result_R_zeta} shows the recall-optimal probabilities $R$ and their lower bounds $\zeta$ of the solution $\mathcal{C}^*$ obtained from the proposed method as blue and magenta lines, respectively, for each sample size $N$ with $\mathcal{S}_{N, \delta, \beta}$ defined by (\ref{S_delta_beta}). $R(\mathcal{M}_\mathbb{P}, \mathcal{C}^*, \mathcal{S}_{N, \delta, \beta})$ is approximately computed by Monte Carlo method in (\ref{recall_opt_prob}) with $1000$ samples. We observe that $\zeta$ bounds $R$ from below and exceeds $\delta_0$ for all sample sizes, and the error between them decreases towards $0$. Moreover, the lower bounds are always obtained as the maximal values $(1- \beta)^{1/N}$. This implies that the proposed method achieves \textbf{(PC1)} in Problem \ref{problem:PACoptimal_diff}. Then, we compare $R$ with the maximum recall-optimal probabilities $R_\mathrm{sub}$ over all proper subsets of the solution $\mathcal{C}^*$ shown in Fig.\ \ref{ex:result_R_zeta} (green lines). $R_\mathrm{sub}$ are always lower than $R$, and thus the obtained solution achieves \textbf{(PC2)} well.
Fig.\ \ref{ex:results_F_eta_max} (resp., \ref{ex:results_F_eta_min}) shows the maximum (resp., minimum) SPR-cause probabilities $F_\mathrm{max}$ (resp., $F_\mathrm{min}$) and the maximum (resp., minimum) lower bounds $\eta_\mathrm{max}$ (resp., $\eta_\mathrm{min}$) over all state sets $C$ in the obtained solution $\mathcal{C}^*$ as blue and green lines, respectively, for each sample size $N$. 
Likewise $R$, $F_\mathrm{max}$ and $F_\mathrm{min}$ are approximately calculated with $1000$ samples.
$\eta_\mathrm{max}$ (resp., $\eta_\mathrm{min}$) bounds $F_\mathrm{max}$ (resp., $F_\mathrm{min}$) from blow for each sample size. The error between $\eta$ and $F$ decreases with sample size.

To investigate the relations between obtained SPR causes and the system behavior, we depict $\mathcal{C}^*$ obtained from the proposed method with $\delta=0.001$ and $N=100$ for the last experiment in Figs.\ \ref{ex:env1} and \ref{ex:env2}, and they correspond to the yellow, orange, and green cells. The sets of the yellow and orange/green cells correspond to SPR causes with $p_1 < p_2$ and $p_1 > p_2$, respectively.
In Fig.\ \ref{ex:env1}, we observed that our method found all possible SPR causes. Moreover, by modifying the one-way aisle in the top middle into a two-way aisle in Fig.\ \ref{ex:env2}, the size of the potential SPR causes is reduced. 

We compare the proposed method with the na\"ive approaches. $\mathrm{NA1}$ only found SPR causes $\{ (0,8), (1,8), (2,7), (4,6) \}$ in Fig.\ \ref{ex:env1}, and could not find SPR causes associated with $p_1 \!>\! p_2$. Thus, the recall-optimal probabilities obtained from $\mathrm{NA1}$ are lower than $1$ when $\delta \!=\! \delta_0$. In contrast, the proposed method further found $\{(3,5), (7,8)\}$ and $\{(5,5), (7,8)\}$, and hence achieves \textbf{(PC1)}. $\mathrm{NA2}$ additionally found SPR causes $\{(3,5), (7,8)\}$ and $\{(5,5), (7,8)\}$ in the environment in Fig.\ \ref{ex:env1}, but removing them does not reduce the recall-optimal probability when $\delta \!=\! \delta_1$. Hence, the solution provided by $\mathrm{NA2}$ is redundant in the sense of (\textbf{PC2}). The proposed method only provides $\{(0,9),\! (1,8),\! (2,7),\! (4,6)\}$ when $\delta = \delta_1$ and thus achieves the nonredundancy.

\section{Related Works}
% 動的システム上の因果探索・推論
\subsection{Formal Verification and Synthesis for Uncertain Systems}
\label{RelatedWorks:VerificationSynthesis}
Formal verification and synthesis for epistemically uncertain systems have been successfully advocated to guarantee that the system satisfies a given specification \citep{cubuktepe2020scenario, badings2022scenario, badings2022sampling, nilim2005robust, wiesemann2013robust, wolff2012robust, rickard2024learning}. The existing works, however, have not considered causal analysis. \cite{cubuktepe2020scenario, badings2022scenario, badings2022sampling} have developed the PAC-guaranteed verification methods for upMDPs and upMCs based on SOP formulation using samples of model parameters. The synthesis of probabilistically guaranteed control policies on upMDPs has also been investigated using optimization and learning based on system behavior samples in \citep{nilim2005robust, wiesemann2013robust, wolff2012robust, rickard2024learning}.
Identifications of uncertain MDPs also have been deeply studied using Bayesian learning \citep{dearden1999model, ross2007bayes, suilen2022robust} and data-driven abstractions \citep{jackson2020safety, lavaei2022constructing}.

\subsection{Causal Analysis for Deterministic and Stochastic Systems}
\label{RelatedWorks:CausalAnalysis}
% \textcolor{red}{ToDo}:
Causal analysis for deterministic and stochastic systems has been intensively investigated \citep{baier2024foundations, coenen2022temporal, finkbeiner2024synthesis, dimitrova2020probabilistic, finkbeiner2023counterfactuals, kazemi2022causal, baier2024backward, oberst2019counterfactual, tsirtsis2021counterfactual, triantafyllou2022actual, baier2021game}. However, existing works have not considered the epistemic uncertainty of the system model itself.
\cite{baier2024foundations} have extended the PR cause-effect relations in \citep{baier2021probabilistic, baier2022probability} based on linear temporal logic. The PR causality and their model checking-based identification have the advantage of an explainable verification for probabilistic systems. \cite{coenen2022temporal, finkbeiner2024synthesis} have considered temporal causality on transition systems using omega-regular properties. An action-level PR causality has been considered on MDPs in \citep{dimitrova2020probabilistic}
based on hyperproperty \citep{clarkson2014temporal}. 
Aside from the PR principle, another well-known causality is based on \textit{counterfactuality}. Counterfactual reasoning operators have recently been incorporated into some temporal logic frameworks in \citep{finkbeiner2023counterfactuals, kazemi2022causal}. 
% In \cite{}, the authors have investigated a counterfactual reasoning-based causal inference using distance functions. 
To perform counterfactual reasoning on MDPs, \cite{oberst2019counterfactual, tsirtsis2021counterfactual, triantafyllou2024agent, triantafyllou2022actual} have conducted encoding of MDPs using structural causal models. 
% In \cite{pawlowski2020deep, crupi2022counterfactual}, the authors have investigated SCM encoding of neural networks. 
As a complementary concept for causality, responsibility attribution has been investigated in \citep{baier2024backward, triantafyllou2022actual, baier2021game}.

\section{Conclusion}
This paper presented a novel approach to exhaustively and nonredundantly identify subsets of states as potential probability-raising (PR) causes on an uncertain parametric Markov decision process.
Based on model checking and a set covering strategy for the sampled parameters, the identified potential PR causes satisfy the following three properties in a probabilistically approximately correct manner: 
(i) The probability of each identified subset being a PR cause is at least a specified value. (ii) A lower bound of the probability that the undesired paths pass through the subsets as much as possible is maximized. (iii) Eliminating any identified subset degrades the lower bound.
We plan to extend the proposed method to concurrent systems.

% References
\bibliography{reference}

\newpage

\onecolumn

\title{Supplementary Material}
\maketitle

\appendix

\section{Derivation of Main Theorems}
\subsection{Scenario Optimization}
\label{appendix:scenario}
In this section, we briefly review the scenario optimization programs (SOPs) \citep{campi2008exact, campi2011sampling}, and how the obtained solutions are guaranteed with PAC bounds. We consider a family of constraints $K^u \subseteq \mathbb{R}^d$ with a natural number $d$ parameterized by $u \in \mathcal{V}$, where $\mathcal{V}$ is an uncountable (and possibly continuous) parameter space, and $K^u$ is convex and closed for each $u$. We assume that $\mathcal{V}$ is endowed with a $\sigma$-algebra $\mathcal{B}(\mathcal{V})$ and a probability measure $\mathbb{P}$. Suppose that $N$ i.i.d. samples sequence $\mathcal{U}_N = (u_i)_{i=1,\ldots,N}$ are available.

For any $\mathcal{U}'_N \subseteq \mathcal{U}_N$, we consider the following optimization program called SOP:
\begin{align}
\label{SOP_template}
    \max_{\gamma} \quad &\gamma \nonumber \\
    \mathrm{s. t.} \quad & \gamma \in  \bigcap_{u \in \mathcal{U}'_{N}} K^u,
\end{align}
We introduce the following assumptions.
\begin{assumption}[Assumption 2.1 in \citep{campi2011sampling}]
\label{assum:exist_unique}
    Every optimization problem subject to a finite set $F \subset \mathcal{V}$, i.e., 
    \begin{align}
    \label{finite_SOP}
        \max_{\gamma} \quad &\gamma \nonumber \\
    \mathrm{s. t.} \quad & \gamma \in  \bigcap_{u \in F} K^u,
    \end{align}
    is feasible, and its feasibility domain has a nonempty interior. Moreover, the solution of (\ref{finite_SOP}) exists and is unique.
\end{assumption}
\begin{assumption}[Assumption 2.2 in \citep{campi2011sampling}]
    \label{assum:violates}
    Almost surely with respect to multi-sample $\mathcal{U}_N$, the solution $\gamma^*$ of the optimization program in (\ref{SOP_template}) violates constraints for all $u \in \mathcal{U}_N \setminus \mathcal{U}'_N$.
\end{assumption}

Then, we have the following arguments.
\begin{theorem}[Theorem 1 in \citep{campi2008exact}]
    \label{thm_scenario_exact}
    Under Assumption \ref{assum:exist_unique}, if $|\mathcal{U}'_N| = 0$, then we have
    \begin{align}
        \mathbb{P}^N \left\{ \mathbb{P}(\gamma^* \not \in K^u) > \varepsilon \right\} < \sum_{i=0}^{d-1} \varepsilon^i (1 - \varepsilon)^{N - i},
    \end{align}
    where $\gamma^*$ is the solution to (\ref{SOP_template}).
\end{theorem}

\begin{theorem}[Theorem 2.1 in \citep{campi2011sampling}]
    \label{thm_scenario_discarding}
    Under Assumptions \ref{assum:exist_unique} and \ref{assum:violates}, we have that
    \begin{align}
        \mathbb{P}^N \left\{ \mathbb{P}(\gamma^* \not \in K^u) > \varepsilon \right\} < \binom{k + d -1}{k}\sum_{i=0}^{k+d-1} \varepsilon^i (1 - \varepsilon)^{N - i},
    \end{align}
    where $\gamma^*$ is the solution to (\ref{SOP_template}).
\end{theorem}

\subsection{Derivation of Theorems \ref{thm_certif} and \ref{thm_optimality}}
\label{Appendix_thm12}
In this section, we prove Theorems \ref{thm_certif} and \ref{thm_optimality} based on scenario optimization program (SOP) formulations.
We first review the following lemma to show the necessary and sufficient condition of SPR causes.
\begin{lemma}[Lemma 21 in \citep{baier2022probability}]
    \label{lemma_SPR_cause_necessary_sufficient}
    For any pMDP $\mathcal{M}_u$, any $C \subset S$, and any $u_j \in \mathcal{V}$, $C$ is an SPR cause on $\mathcal{M}_{u_j}$ if and only if $C \subseteq \mathscr{C}_{S, u_j}$ and $C$ satisfies \textbf{(M)} in Def.\ref{def:SPR_cause}, where $\mathscr{C}_{S, u_j}$ is defined by (\ref{all_cause_states}).
\end{lemma}

For any MDP $\mathcal{M}_u$ and any $C \subset S$, we define $\gamma^u_C$ as 
\begin{align}
    \gamma^u_C = 
    \begin{cases}
        1 & \mbox{ if } C \mbox{ is an SPR cause on $\mathcal{M}_u$},\\
        0 & \mbox{ otherwise},
    \end{cases}
    % \min_{c \in C} \min_{\pi} \mathrm{Pr}_{\mathcal{M}_u}^\pi (\Diamond E \;|\; \neg C \mathrm{U} c) - \mathrm{Pr}_{\mathcal{M}_u}^\pi(\Diamond E).
\end{align}
% Clearly, $\gamma^u_C > 0 $ holds if and only if $C$ is an SPR cause on $\mathcal{M}_u$.
For any $C \subseteq S$, let $K^u_C = [0, \gamma^u_C]$. For any sequence of $N$ sampled parameters $\mathcal{U}_N \in \mathcal{V}^N$, we define $\mathcal{U}^C_{N,1} = \{ u \in \mathcal{U}_N \;|\; C \subseteq \mathscr{C}_{S,u}, (\mathcal{M}_\mathbb{P}, C) \models \textbf{(M)} \}$, where we treat the same elements in $\mathcal{U}_N$ as duplicates. Then, we define $\mathcal{U}^C_{N,0} = \mathcal{U}_N \setminus \mathcal{U}^C_{N,1}$, where $\mathscr{C}_{S,u}$ is defined as (\ref{all_cause_states}).
Then, we consider the following SOP:

\begin{align}
\label{SOP_SPR_cause}
    \max_{\gamma} \quad &\gamma \nonumber \\
    \mathrm{s. t.} \quad & \gamma \in  \bigcap_{u \in \mathcal{U}'_N} K^u_C,
    & 
\end{align}
where $\mathcal{U}'_N \subseteq \mathcal{U}_N$.
% Note that the solution $\gamma^*_C = \min_{u \in \mathcal{U}^C_{N, 1}} \gamma^u_C$ of (\ref{SCP}) is positive.
Let $\gamma^*_C$ be the solution of (\ref{SOP_SPR_cause}). Note that $\gamma^*_C = \min_{u \in \mathcal{U}'_N} \gamma^u_C$.

\begin{proof}[proof of Theorem \ref{thm_certif}]
For any $C \subset S$, we consider the two cases (1) $\mathcal{U}^C_{N,0} = \emptyset$ and (2) $\mathcal{U}^C_{N,0} \neq \emptyset$.

(1): We set $\mathcal{U}'_N = \mathcal{U}^C_{N,1}$. By Lemma \ref{lemma_SPR_cause_necessary_sufficient}, for any $u \in \mathcal{U}_{N, 1}^C$, $C$ is an SPR cause on $\mathcal{M}_u$. Thus, the intersection of all $K^u_C$ with $u \in \mathcal{U}^C_{N, 1}$ is $[0,1]$ and hence there exists a unique solution $\gamma^*_C = 1$ to (\ref{SOP_SPR_cause}). Thus, for any $\varepsilon > 0$, by Theorem \ref{thm_scenario_exact}, we have
\begin{align}
    \mathbb{P}^N \left\{ \mathbb{P}(\gamma^*_C \not\in K^u_C) > \varepsilon \right\} < (1 - \varepsilon)^N,
\end{align}
and hence
\begin{align}
    \mathbb{P}^N \left\{ \mathbb{P}(\gamma^*_C \in K^u_C) \geq 1 - \varepsilon \right\} \geq 1 - (1 - \varepsilon)^N,
\end{align}
where the internal probability $\mathbb{P}$ measures the frequency of $u$ for $K^u_C$.
We choose $\beta = 1 - (1 - \varepsilon)^N$ and hence
\begin{align}
    \mathbb{P}^N \left\{ \mathbb{P}(\gamma^u_C \geq \gamma^*_C) \geq (1 - \beta)^{1/N} \right\} \geq \beta.
\end{align}
We have that $ \gamma^*_C = 1$ and thus 
\begin{align}
    \mathbb{P}^N \left\{ \mathbb{P}(\gamma^u_C = 1) \geq (1 - \beta)^{1/N} \right\} \geq \beta.
\end{align}
% Moreover, $\mathcal{M}_\mathbb{P}$ is graph-preserving and hence $(\mathcal{M}_u, C) \models \mbox{(M)}$ for all $u \in \mathcal{V}$ when $\mathcal{U}_{N, 1}^C \neq \emptyset$. 
Therefore, we have
\begin{align}
    \mathbb{P}^N \left\{ F(\mathcal{M}_\mathbb{P}, C) \geq (1- \beta)^N \right\} \geq \beta.
\end{align}

(2): For any $\mathcal{U}'_N \neq \emptyset$ in (\ref{SOP_SPR_cause}), the intersection of all $K^u_C$ with $u \in \mathcal{U}'_N$ is nonempty and there exists a unique solution $\gamma^*_C$ to (\ref{SOP_SPR_cause}). Thus, for any $\varepsilon > 0$ and any a-priori number of discarded sampled parameters $q = |\mathcal{U}_N \setminus \mathcal{U}'_N|$, by Theorem \ref{thm_scenario_discarding}, we have
\begin{align}
\label{eq_PACbound_case2}
    \mathbb{P}^N \left\{ \mathbb{P}(\gamma^*_C(q) \not\in K^u_C) \leq \varepsilon \right\} & = \mathbb{P}^N \left\{ \mathbb{P} (\gamma^*_C(q) \in K^u_C) \geq 1 - \varepsilon \right\} \nonumber \\
    & \geq 1 - \sum_{i=0}^{q} \binom{N}{i} \varepsilon^i (1 - \varepsilon)^{N-i}.
\end{align}
We set $t = 1 - \varepsilon$ and equate the right-hand side of (\ref{eq_PACbound_case2}) to $1 - (1- \beta)/N$. Then, we have
\begin{align}
\label{q_bound_for_F}
    \mathbb{P}^N \left\{ \mathbb{P} (\gamma^u_C \geq \gamma^*_C(q)) \geq t^*(q) \right\} \geq 1 - \frac{1 - \beta}{N},
\end{align}
where $t^*(q)$ is the solution of
\begin{align}
        \frac{1 - \beta}{N} = \sum_{i=0}^{q} \binom{N}{i} (1 - t)^i t^{N-i}.
\end{align}
We consider the case $q = |\mathcal{U}^C_{N,0}|$ to derive the lower bound of $F(\mathcal{M}_\mathbb{P}, C)$ for $\mathcal{U}'_N = \mathcal{U}^C_{N,1}$. Since the number $|\mathcal{U}_{N,0}^C|$ is not known before observing actual samples, we simultaneously consider the bound for all $q=1, \ldots, N-1$. By Bool's inequality and (\ref{q_bound_for_F}), we have
\begin{align}
    \mathbb{P}^N \left\{ \cap_{q=1}^{N} \mathcal{E}_q
\right\} & = 1 - \mathbb{P}^N \left\{ \cup_{q=1}^{N} \mathcal{E}_q^c \right\} \nonumber \\
         & \geq 1 - \sum_{q=1}^{N} \mathbb{P}^N \left\{ \mathcal{E}_q^c \right\} \nonumber \\
         & \geq \beta,
\end{align}
where $\mathcal{E}_q$ denotes the event $\mathbb{P} (\gamma^u_C \geq \gamma^*_C(q)) \geq t^*(q)$ and $\mathcal{E}_q^c$ is its complement event.
Clearly, we have
\begin{align}
\label{beta_bound_case2}
    \mathbb{P}^N \left\{ \mathbb{P} (\gamma^u_C \geq \gamma^*_C) \geq t^*(|\mathcal{U}^C_{N,0}|) \right\} & \geq \mathbb{P}^N \left\{ \forall q'=1,\ldots,N, \mathbb{P} (\gamma^u_C \geq \gamma^*_C) \geq t^*(q') \right\} \nonumber \\
    & \geq \beta. 
\end{align}
Note that we can take $\mathcal{U}'_N = \mathcal{U}^C_{N, 1}$ when $q = |\mathcal{U}^C_{N,0}|$. By Lemma \ref{lemma_SPR_cause_necessary_sufficient}, for any $u \in \mathcal{U}_{N, 1}^C$, $C$ is an SPR cause on $\mathcal{M}_u$. Thus, the intersection of all $K^u_C$ with $u \in \mathcal{U}^C_{N, 1}$ is $[0,1]$ and hence there exists a unique solution $\gamma^*_C = 1$ to (\ref{SOP_SPR_cause}) when $\mathcal{U}'_N = \mathcal{U}^C_{N,1}$. Thus, we have
\begin{align}
    \mathbb{P}^N \left\{ \mathbb{P} (\gamma^u_C \geq \gamma^*_C) = t^* (|\mathcal{U}^C_{N,0}|) \right\}
    & = \mathbb{P}^N \left\{ \mathbb{P} (\gamma^u_C = 1) \geq t^*(|\mathcal{U}^C_{N,0}|) \right\} \nonumber \\
    & = \mathbb{P}^N \left\{ F(\mathcal{M}_\mathbb{P}, C) \geq t^*(|\mathcal{U}^C_{N,0}|) \right\}.
\end{align}
Therefore, by (\ref{beta_bound_case2}), we have
\begin{align}
    \mathbb{P}^N \left\{ F(\mathcal{M}_\mathbb{P}, C) \geq t^*(|\mathcal{U}^C_{N,0}|) \right\} \geq \beta.
\end{align}
\end{proof}

We show the following lemmas before proving Theorem \ref{thm_optimality}.
% \mathcal{C}_{\mathcal{U}_n, k} contains S/S_dis の補題
\begin{comment}
\begin{lemma}
    For any upMDP $\mathcal{M}_\mathbb{P}$, any set $\mathcal{U}_N$ of $N \geq 1$ sampled parameters, an dany $k > 0$, we have
    \begin{align}
        S \setminus S_{dis} \subseteq \mathcal{C}_{\mathcal{U}_N, k},
    \end{align}
    where $\mathcal{C}_{\mathcal{U}_N, k}$ is defined as (\ref{capitalC_UN_k}) and $S_{dis}$ is defined as (\ref{S_dis})
\end{lemma}
\end{comment}
For any path property $\varphi$ and any path $\rho$ of an MDP, we denote $\rho \models \varphi$ if and only if $\rho$ satisfies $\varphi$. Moreover, recall that we denote all finite paths on the pMDP as $\mathrm{FinPath}$.
\begin{lemma}
\label{lemma_path_implication}
    For any MDP $\mathcal{M}_u$, any $S' \subseteq S$, any SPR cause $C'$ on $\mathcal{M}_u$ that is a subset of $S'$, we have that
    \begin{align}
    % \label{C_star_u_condition1}
        % & \rho \models \neg \Diamond C \land \Diamond E \implies \rho \models \neg \Diamond C' \land \Diamond E, \\
    \label{C_star_u_conditoin2}
        & \forall \rho \in \mathrm{FinPath}, \rho \models \Diamond C' \implies \Diamond C^*_{S',u},
    \end{align}
    where $C^*_{S',u}$ is defined by (\ref{canonical_cause}).
\end{lemma}
\begin{proof}
    Suppose that $\rho \models \neg \Diamond C^*_{S',u} \land \Diamond C'$. This implies that there exists $s \in C'$ such that $\rho \models \neg C^*_{S',u} \mathrm{U} s$. 
    By (\ref{canonical_cause}), for any $c \in \mathscr{C}_{S', u} \setminus C^*_{S', u}$ and for any $\rho \in \mathrm{FinPath}$, we have $\rho \not\models \neg \mathscr{C}_{S', u} \mathrm{U} c$. Thus, there exist $c' \in \mathscr{C}_{S', u}$ with $c' \neq c$ and a path $\rho'$ such that $\rho' \models \Diamond ( c' \land \Diamond c)$ and $\rho' \models \neg \mathscr{C}_{S', u} \mathrm{U} c'$. This implies that, for any path $\rho$, we have $\rho \not\models \neg C^*_{S', u} \mathrm{U} c$ since $C^*_{S',u}$ is constructed from (\ref{canonical_cause}). However, this contradicts the assumption.
\end{proof}

\begin{lemma}
\label{lemma_recall_inequality}
    For any MDP $\mathcal{M}_u$, any $S' \subset S$, and any SPR cause $C \subseteq S'$, $C^*_{S',u}$ is recall optimal over $\mathcal{S}$ for any $\mathcal{S} \subseteq 2^{S'}$,
    where $C^*_{S', u}$ is defined as (\ref{canonical_cause}).
\end{lemma}
\begin{proof}
    By Lemma \ref{lemma_path_implication}, for any policy $\pi$ and any SPR cause $C' \in \mathcal{S}$ on $\mathcal{M}_u$, we have that
    \begin{align}
        % & \mathrm{Pr}^\pi_{\mathcal{M}_u}(\neg \Diamond C' \land \Diamond E) \geq \mathrm{Pr}^\pi_{\mathcal{M}_u}(\neg \Diamond C \land \Diamond E), \\
        & \mathrm{Pr}^\pi_{\mathcal{M}_u}(\Diamond C' \land \Diamond E) \leq \mathrm{Pr}^\pi_{\mathcal{M}_u}(\Diamond C^*_{S',u} \land\Diamond E).
    \end{align} 
    This yields the claim.
\end{proof}

\begin{lemma}
\label{lemma_recall_sufficient}
    For any MDP $\mathcal{M}_u$, any $S' \subseteq S$, and any $C \subseteq \mathscr{C}_{S', u}$ that satisfies \textbf{(M)} in Def. \ref{def:SPR_cause}, if (\ref{recall_sufficient}) holds, then $C$ is recall-optimal on $\mathcal{M}_{u}$ over $\mathcal{S} = 2^{S'}$.
    \begin{align}
    \label{recall_sufficient}
        & \forall \rho \in \mathrm{FinPath}, \rho \models \Diamond C^*_{S', u} \land \Diamond E \implies \Diamond C,
    \end{align}
    where $\mathscr{C}_{S',u}$ is defined by (\ref{all_cause_states}), and $C^*_{S',u}$ is defined by (\ref{canonical_cause}).
\end{lemma}
\begin{proof}
    By (\ref{recall_sufficient}), we have
    \begin{align}
    \label{recall_inequality}
        \forall \pi, \mathrm{Pr}^\pi_{\mathcal{M}_{u_j}}(\Diamond C^*_{S', u_j} \;|\; \Diamond E) \leq \mathrm{Pr}^\pi_{\mathcal{M}_{u_j}}(\Diamond C \;|\; \Diamond E).
    \end{align}
    By Lemma \ref{lemma_recall_inequality}, $C^*_{S', u_j}$ is recall optimal on $\mathcal{M}_{u_j}$ over $\mathcal{S}$. Thus, by (\ref{recall_inequality}), for any SPR cause $C' \subseteq S'$, we have
    \begin{align}
        \forall \pi, \mathrm{Pr}^\pi_{\mathcal{M}_{u_j}}(\Diamond C' \;|\; \Diamond E) \leq \mathrm{Pr}^\pi_{\mathcal{M}_{u_j}}(\Diamond C^*_{S', u_j} \;|\; \Diamond E) \leq \mathrm{Pr}^\pi_{\mathcal{M}_{u_j}}(\Diamond C \;|\; \Diamond E).
    \end{align}
    By Lemma \ref{lemma_SPR_cause_necessary_sufficient}, $C$ is SPR cause on $\mathcal{M}_{u_j}$.
    Therefore, $C$ is recall optimal on $\mathcal{M}_{u_j}$ over $\mathcal{S}$.
\end{proof}

For any MDP $\mathcal{M}_u$ and any $\mathcal{C}, \mathcal{S} \subseteq 2^S$, we define ${\xi}^u_{\mathcal{C}, \mathcal{S}}$ as
\begin{align}
\label{xi_u_CS}
    {\xi}^u_{\mathcal{C}, \mathcal{S}} =
    \begin{cases}
        1 & \mbox{ if } \exists C \in \mathcal{C} \mbox{ s.t. } C \mbox{ is recall-optimal on $\mathcal{M}_u$ over $\mathcal{S}$ },\\
        0 & \mbox{ otherwise}.
    \end{cases}
\end{align}
% where $\tilde{\xi}^u_{C, \mathcal{S}}$ is defined as
% \begin{align}
    % \tilde{\xi}^u_{C, \mathcal{S}} =
      % \min_{C' \in \mathcal{S}} \mathrm{Pr}^\mathrm{min}_{\mathcal{M}_u}(\Diamond C \;|\; \Diamond E) - \mathrm{Pr}^\mathrm{min}_{\mathcal{M}_u}(\Diamond C' \;|\; \Diamond E),
% \end{align}
% Clearly, ${\xi}^u_{\mathcal{C}, \mathcal{S}}$ is nonnegative if and only if there exists an SPR cause $C \in \mathcal{C}$ such that $C$ satisfies the condition (\ref{PR-recall_optimal}) for $\mathcal{S}$.
For any $\mathcal{C}, \mathcal{S} \subseteq 2^S$, we define the subset of sampled parameters $\mathcal{U}^{\mathcal{C}, \mathcal{S}}_{N, 1}$ as
\begin{align}
\label{U_CS'_N1}
    \mathcal{U}^{\mathcal{C}, \mathcal{S}}_{N, 1} = \{ u \in \mathcal{U}_N \;|\; &\exists C \in \mathcal{C} \mbox{ s.t. }
    C \subseteq \mathscr{C}_{S',u}, (\mathcal{M}_u, C) \models \textbf{(M)}, \nonumber \\
    &\forall \rho \in \mathrm{FinPath}, \rho \models \Diamond C^*_{S',u} \land \Diamond E \implies \Diamond C\},
\end{align}
where, we treat the same elements in $\mathcal{U}_N$ as duplicates, and $\mathscr{C}_{S', u}$ and $C^*_{S',u}$ are defined by (\ref{all_cause_states}) and (\ref{canonical_cause}) with $S' = \cup_{C' \in \mathcal{S}} C'$, respectively. Then, we define $\mathcal{U}_{N, 0}^{\mathcal{C}, \mathcal{S}}$ as $\mathcal{U}_N \setminus \mathcal{U}_{N,1}^{\mathcal{C}, \mathcal{S}}$.
To prove Theorem \ref{thm_optimality}, we consider the following SOP:
\begin{align}
\label{SCP_for_thm2}
    \max_{{\xi}} \quad & {\xi} \nonumber \\
    \mathrm{s. t.} \quad & {\xi} \in  \bigcap_{u \in \mathcal{U}'_N} {J}^u_{\mathcal{C}, \mathcal{S}},
\end{align}
where ${J}^u_{\mathcal{C}, \mathcal{S}} = [0, {\xi}^u_{\mathcal{C}, \mathcal{S}}]$ and $\mathcal{U}'_N \subseteq \mathcal{U}_N$. We denote the solution of (\ref{SCP_for_thm2}) by ${\xi}^*_{\mathcal{C}, \mathcal{S}}$. Note that ${\xi}^*_{\mathcal{C}, \mathcal{S}} = \min_{u \in \mathcal{U}'_N} {\xi}^u_{\mathcal{C}, \mathcal{S}}$.
\begin{proof}[Proof of Theorem \ref{thm_optimality}]
    For any $\mathcal{C}, \mathcal{S} \subset 2^S$, we consider the two case (1) $\mathcal{U}_{N,0}^{\mathcal{C}, \mathcal{S}} = \emptyset$ and (2) $\mathcal{U}_{N,0}^{\mathcal{C}, \mathcal{S}} \neq \emptyset$.

    (1): We set $\mathcal{U}'_N = \mathcal{U}^{\mathcal{C}, \mathcal{S}}_{N, 1}$. By Lemma \ref{lemma_recall_sufficient}, for any $u \in \mathcal{U}'_N$, there exists $C \in \mathcal{C}$ such that $C$ is recall-optimal on $\mathcal{M}_u$ over $\mathcal{S}$. Thus, the intersection of all $J^u_{\mathcal{C}, \mathcal{S}}$ with $u \in \mathcal{U}'_N$ is $[0,1]$ and hence there exists a unique solution $\xi^*_{\mathcal{C},\mathcal{S}} = 1$ to (\ref{SOP_SPR_cause}). Thus, for any $\varepsilon > 0$, by Theorem \ref{thm_scenario_exact}, we have
    \begin{align}
        \mathbb{P}^N \left\{ \mathbb{P} (\xi^*_{\mathcal{C}, \mathcal{S}} \in J^u_{\mathcal{C}, \mathcal{S}}) \geq 1 - \varepsilon \right\} \geq 1 - (1 - \varepsilon)^N,
    \end{align}
    where the internal probability $\mathbb{P}$ measures the frequency of $u$ for $J^u_{\mathcal{C}, \mathcal{S}}$.
    We choose $\beta = 1 - (1 - \varepsilon)^N$ and hence
    \begin{align}
        \mathbb{P}^N \left\{ \mathbb{P}(\xi^u_{\mathcal{C}, \mathcal{S}} \geq \xi^*_{\mathcal{C}, \mathcal{S}}) \geq (1 - \beta)^{1/N} \right\} \geq \beta.
    \end{align}
    We have that $\xi^*_{\mathcal{C}, \mathcal{S}} = 1$ and thus
    \begin{align}
        \mathbb{P}^N \left\{ \mathbb{P}(\xi^u_{\mathcal{C}, \mathcal{S}} = 1) \geq (1 - \beta)^{1/N} \right\} \geq \beta.
    \end{align}
    Therefore, we have
    \begin{align}
        \mathbb{P}^N \left\{ R(\mathcal{M}_\mathbb{P}, \mathcal{C}, \mathcal{S}) \geq (1 - \beta)^{1/N} \right\} \geq \beta.
    \end{align}

    (2): For any $\mathcal{U}'_N \neq \emptyset$ in (\ref{SCP_for_thm2}), the intersection of all $J^u_{\mathcal{C}, \mathcal{S}}$ with $u \in \mathcal{U}'_N$ is nonempty and there exists a unique solution $\xi^*_{\mathcal{C}, \mathcal{S}}$ to (\ref{SOP_SPR_cause}). Thus, for any $\varepsilon > 0$ and any a-priori number of discarded sampled parameters $q = |\mathcal{U}_N \setminus \mathcal{U}'_N|$, by Theorem \ref{thm_scenario_discarding}, we have
    \begin{align}
    \label{eq_recall_opt_PACbound_case2}
        \mathbb{P}^N \left\{ \mathbb{P}(\xi^*_{\mathcal{C}, \mathcal{S}}(q) \in J^u_{\mathcal{C}, \mathcal{S}}) \geq 1- \varepsilon \right\} \geq 1 - \sum_{i=0}^{q} \binom{N}{i} \varepsilon^i (1 - \varepsilon)^{N - i}.
    \end{align}
    We set $t = 1 - \varepsilon$ and equate the right-hand side of (\ref{eq_recall_opt_PACbound_case2}) to $1 - (1 - \beta)/N$, where $\beta \in (0,1)$. Then, we have
    \begin{align}
    \label{q_bound_for_R}
        \mathbb{P}^N \left\{ \mathbb{P}(\xi^u_{\mathcal{C}, \mathcal{S}} \geq \xi^*_{\mathcal{C}, \mathcal{S}}(q)) \geq t^*(q) \right\} \geq 1 - \frac{1-\beta}{N},
    \end{align}
    where $t^*(q)$ is the solution of
    \begin{align}
    \label{t_solution}
        \frac{1 - \beta}{N} = \sum_{i=0}^{q} \binom{N}{i} (1 - t)^i t^{N-i}.
    \end{align}
    We consider the case $q = |\mathcal{U}^C_{N,0}|$ to derive the lower bound of $R(\mathcal{M}_\mathbb{P}, \mathcal{C}, \mathcal{S})$ for $\mathcal{U}'_N = \mathcal{U}^{\mathcal{C}, \mathcal{S}}_{N,1}$. Since the number $|\mathcal{U}_{N,0}^{\mathcal{C}, \mathcal{S}}|$ is not known before observing actual samples, we consider simultaneously the bound for all $q=1, \ldots, N$. By Bool's inequality and (\ref{q_bound_for_R}), we have
    \begin{align}
        \mathbb{P}^N \left\{ \cap_{q=1}^{N} \mathcal{E}_q
    \right\} & = 1 - \mathbb{P}^N \left\{ \cup_{q=1}^{N} \mathcal{E}_q^c \right\} \nonumber \\
             & \geq 1 - \sum_{q=1}^{N} \mathbb{P}^N \left\{ \mathcal{E}_q^c \right\} \nonumber \\
             & \geq \beta,
    \end{align}
    where $\mathcal{E}_q$ denotes the event $\mathbb{P} (\xi^u_{\mathcal{C}, \mathcal{S}} \geq \xi^*_{\mathcal{C}, \mathcal{S}}) \geq t^*(q)$ and $\mathcal{E}_q^c$ is its complement event.
    Clearly, we have
    \begin{align}
    \label{beta_bound_case2_thm2}
        \mathbb{P}^N \left\{ \mathbb{P} (\xi^u_{\mathcal{C}, \mathcal{S}} \geq \xi^*_{\mathcal{C}, \mathcal{S}}) \geq t^*(|\mathcal{U}^{\mathcal{C},\mathcal{S}}_{N,0}|) \right\}
        & \geq \mathbb{P}^N \left\{ \forall q'=1,\ldots,N, \mathbb{P} (\xi^u_{\mathcal{C}, \mathcal{S}} \geq \xi^*_{\mathcal{C}, \mathcal{S}}) \geq t^*(q') \right\} \nonumber \\
        & \geq \beta. 
    \end{align}
    Note that we can take $\mathcal{U}'_N = \mathcal{U}^C_{N, 1}$ when $q = |\mathcal{U}^C_{N,0}|$. As in the case (1), by Lemma 21 in \citep{baier2022probability} and Lemma \ref{lemma_recall_inequality}, $\xi^*_{\mathcal{C}, \mathcal{S}} = 1$ holds. 
    Thus, we have
    \begin{align}
        \mathbb{P}^N \left\{ \mathbb{P} (\xi^u_{\mathcal{C}, \mathcal{S}} \geq \xi^*_{\mathcal{C}, \mathcal{S}}) = t^* (|\mathcal{U}^C_{N,0}|) \right\} 
        & = \mathbb{P}^N \left\{ \mathbb{P} (\xi^u_{\mathcal{C}, \mathcal{S}} = 1) \geq t^*(|\mathcal{U}^{\mathcal{C}, \mathcal{S}}_{N,0}|) \right\} \nonumber \\
        & = \mathbb{P}^N \left\{ R(\mathcal{M}_\mathbb{P}, \mathcal{C}, \mathcal{S}) \geq t^*(|\mathcal{U}^{\mathcal{C}, \mathcal{S}}_{N,0}|) \right\}.
    \end{align}
    Therefore, by (\ref{beta_bound_case2_thm2}), we have
    \begin{align}
        \mathbb{P}^N \left\{ R(\mathcal{M}_\mathbb{P}, \mathcal{C}, \mathcal{S}) \geq t^*(|\mathcal{U}^{\mathcal{C}, \mathcal{S}}_{N,0}|) \right\} \geq \beta.
    \end{align}
\end{proof}

\subsection{Derivation of Theorem \ref{thm_solution_certif}}
\label{Appendix_thm3}
In this section, we prove Theorem \ref{thm_solution_certif}. First, we show the following lemmas. 
Lemma \ref{lemma_recall_optimal_necessary_sufficient} states that the necessary and sufficient condition for any subset of states to be recall-optimal on any MDP $\mathcal{M}_u$.
Then, Lemma \ref{lemma_recall_optimality_max} states that the solution defined by (\ref{C_star}) accomplishes the maximum lower bound for the recall-optimal probability, that is the solution satisfies \textbf{(PC1)} in Problem \ref{problem:PACoptimal_diff}, based on Lemma \ref{lemma_recall_optimal_necessary_sufficient}.
% Lemma \ref{lemma_recall_optimality_max} guarantees that the solution defined by (\ref{C_star}) satisfies (PC1) in Problem \ref{problem:PACoptimal_diff}. 
Proposition \ref{prop_certify_PC2} with Lemmas \ref{lemma_recall_optimality_max} and \ref{lemma_recall_optimality_diff} assures that $\mathcal{C}^*$ defined by (\ref{C_star}) fulfills \textbf{(PC2)} in Problem \ref{problem:PACoptimal_diff}.

For simplicity, we abbreviate $C^*_{S_N, u_i}$ defined by (\ref{canonical_cause}) and $\mathscr{C}_{S_N, u_i}$ defined by (\ref{all_cause_states}) with $S_N = \cup_{C' \in \mathcal{S}_{N, \delta, \beta}} C'$ as $C^*_i$ and $\mathscr{C}_i$, respectively. 
% We define $\mathcal{C}$ as the set of $C^*_{i}$, that is,
% \begin{align}
% \label{C_max}
    % \mathcal{C} = \{ C^*_i\;|\; & u_i \in \mathcal{U}_N \}.
% \end{align}
Moreover, recall that we denote all finite paths on the pMDP as $\mathrm{FinPath}$.
For any $C, C', \mathscr{C} \subseteq S$, we consider the following condition:
\begin{align}
\label{partial_order_cause}
    \left( C \subseteq \mathscr{C} \right) \land \left( \forall \rho \in \mathrm{FinPath}, \rho \models \Diamond C' \land \Diamond E \implies \Diamond C \right).
\end{align}
% where $\rho$ is a path in $S(AS)^*$.
% We especially denote (\ref{partial_order_cause}) as $\varphi(u_i \succeq u_j)$ if $C = C^*_i \cap S'$ and $C' = C^*_j \cap S'$ for some $S' \subset S$, and $\mathscr{C} = \mathscr{C}_{u_j}$, where $C^*_i$ and $\mathscr{C}_{u_j}$ are defined by (\ref{canonical_cause}) and (\ref{all_cause_states}), respectively.
We denote the first operand in (\ref{partial_order_cause}) as $\varphi_\mathrm{cause}(C, \mathscr{C})$ and the second operand in (\ref{partial_order_cause}) as $\varphi_\mathrm{recall}(C, C')$.

\begin{lemma}
\label{lemma_recall_optimal_necessary_sufficient}
    For any pMDP $\mathcal{M}_u$, any $S' \subset S$, any $C \subseteq S'$ that satisfies \textbf{(M)} in Def. \ref{def:SPR_cause}, and any $u_j \in \mathcal{V}$, $C$ is recall-optimal on $\mathcal{M}_{u_j}$ over $\mathcal{S} = 2^{S'}$ if and only if $\varphi_\mathrm{cause}(C, \mathscr{C}_{S', u_j}) \land \varphi_\mathrm{recall}(C, C^*_{S', u_j})$ holds, where $\mathscr{C}_{S', u_j}$ is defined by (\ref{all_cause_states}) and $C^*_{S', u_j}$ is defined by (\ref{canonical_cause}). 
\end{lemma}
\begin{proof}
\begin{description}
    \item [$(\Leftarrow)$] This immediately follows from Lemma \ref{lemma_recall_sufficient}.
    \item [$(\Rightarrow)$] Suppose that $\neg \varphi_\mathrm{cause}(C, \mathscr{C}_{S', u_j}) \lor \neg \varphi_\mathrm{recall}(C, C^*_{S',u_j})$ holds. Then, we consider two cases (i) $\varphi_\mathrm{cause}(C, \mathscr{C}_{S', u_j})$ does not hold and (ii) $\varphi_\mathrm{cause}(C, \mathscr{C}_{S', u_j})$ holds but $\varphi_\mathrm{recall}(C, C^*_{S',u_j})$ does not hold.
    \begin{description}
        \item [(i)] By Lemma \ref{lemma_SPR_cause_necessary_sufficient}, $C$ is not SPR cause on $\mathcal{M}_{u_j}$, and hence $C$ is not recall optimal $\mathcal{M}_{u_j}$ over $\mathcal{S}$.
        \item [(ii)] By $\neg \varphi_\mathrm{recall}(C, C^*_{S', u_j})$, we have
        \begin{align}
        \label{path_baypass_C}
            \exists \rho \mbox{ s.t. } \rho \models \Diamond C^*_{S',u_j} \land \Diamond E \land \neg \Diamond C.
        \end{align}
        By $\varphi_\mathrm{cause}(C, \mathscr{C}_{S', u_j})$ and Lemma \ref{lemma_SPR_cause_necessary_sufficient}, $C$ is an SPR cause on $\mathcal{M}_u$, and hence, by Lemma \ref{lemma_recall_sufficient},
        \begin{align}
        \label{path_set_inclusion}
            \{ \rho \in \mathrm{FinPath} \;|\; \rho \models \Diamond C^*_{S',u_j} \land \Diamond E \} \supseteq \{ \rho \in \mathrm{FinPath} \;|\; \rho \models \Diamond C \land \Diamond E \}.
        \end{align}
        Thus, by (\ref{path_set_inclusion}) and (\ref{path_baypass_C}), we have
        \begin{align}
            \{ \rho \in \mathrm{FinPath} \;|\; \rho \models \Diamond C^*_{S',u_j} \land \Diamond E \} \supset \{ \rho \in \mathrm{FinPath} \;|\; \rho \models \Diamond C \land \Diamond E \}.
        \end{align}
        Hence, we have
        \begin{align}
            \exists \pi \mbox{ s.t. } \mathrm{Pr}^\pi_{\mathcal{M}_{u_j}}(\Diamond C^*_{S',u_j} \;|\; \Diamond E) > \mathrm{Pr}^\pi_{\mathcal{M}_{u_j}}(\Diamond C \;|\; \Diamond E).
        \end{align}
        Therefore, $C$ is not recall optimal $\mathcal{M}_{u_j}$ over $\mathcal{S}$.
    \end{description}
\end{description}
\end{proof}

\begin{lemma}
\label{lemma_recall_optimality_max}
    For any upMDP $\mathcal{M}_\mathbb{P}$, any $\delta, \beta \in [0,1]$, and any $N > 0$, $\mathcal{C}^*$ defined by (\ref{C_star}) satisfies the following two conditions:
    \begin{align}
        \label{complete_recall_optimality_Cstar}
        \forall u \in \mathcal{U}_N, \exists C \in \mathcal{C}^* \mbox{ s.t. } C \mbox{ is recall-optimal on } \mathcal{M}_u \mbox{ over } \mathcal{S}_{N, \delta, \beta},
    \end{align}
    and
    \begin{align}
    \label{bound_tightness_Cstar}
        {\zeta}_N(\mathcal{C}^*, \mathcal{S}_{N, \delta, \beta}, \beta) = (1 - \beta)^{1/N},
    \end{align}
    where $\zeta_N$ is defined as (\ref{zeta}).
\end{lemma}

\begin{proof}
By (\ref{C_star}), for any $u_j \in \mathcal{U}_N$, there exists $C_i \in \mathcal{C}^*$ such that $C_i \subseteq \cup_{C' \in \mathcal{S}_{N, \delta, \beta}} C'$, $C$ satisfies \textbf{(M)}, and $\varphi_\mathrm{cause}(C_i, \mathscr{C}_{j}) \land \varphi_\mathrm{recall}(C_i, C^*_j)$ holds. Hence, by Lemma \ref{lemma_recall_optimal_necessary_sufficient}, $C_i$ is recall-optimal on $\mathcal{M}_{u_j}$ over $\mathcal{S}_{N, \delta, \beta}$.
Thus, we have (\ref{complete_recall_optimality_Cstar}) and
\begin{align}
    m(\mathcal{C}^*, \mathcal{S}_{N,\delta,\beta}, \mathcal{U}_N) = |\mathcal{U}_N| = N,
\end{align}
where $m$ is he number of samples in $\mathcal{U}_N$ such that there exists $C \in \mathcal{C}^*$ that satisfies condition defined by(\ref{m_count}).
Therefore, since the maximum value of $t^*$ defined by (\ref{t_star}) is $t^*(0,\beta) = (1 - \beta)^{1/N}$ for any $\beta$ and any $N$, (\ref{bound_tightness_Cstar}) holds by (\ref{zeta}) and (\ref{complete_recall_optimality_Cstar}).
% 定義よりC^*はCを通ってEに行くパスを全てカバーするので，C^*もまたrecall最適なものを任意のu in U_Nに対して持つ．
\end{proof}

By Lemma \ref{lemma_recall_optimality_max}, $\mathcal{C}^*$ defined by (\ref{C_star}) achieves the condition \textbf{(PC1)} in Problem \ref{problem:PACoptimal_diff} since the maximum value of $\zeta_N$ defined in (\ref{zeta}) is $(1-\beta)^{1/N}$.
\begin{lemma}
\label{lemma_recall_optimality_diff}
    For any upMDP $\mathcal{M}_\mathbb{P}$, any $\mathcal{C}, \mathcal{S} \subseteq 2^S$, and any $\mathcal{C}' \subseteq \mathcal{C}$, we have the following equation:
    \begin{align}
        R(\mathcal{M}_\mathbb{P}, \mathcal{C}, \mathcal{S}) - R(\mathcal{M}_\mathbb{P}, \mathcal{C}', \mathcal{S}) = \int_{\mathcal{V}} \xi_{\mathrm{diff}}(u, \mathcal{C}, \mathcal{C}', \mathcal{S}) \mathrm{d}\mathbb{P}(u),
    \end{align}
    where $\xi^\mathrm{diff} : \mathcal{V} \times 2^{2^S} \times 2^{2^S} \times 2^{2^S}$ is defined as
    \begin{align}
    \label{xi_diff}
        \xi_\mathrm{diff}(u, \mathcal{C}, \mathcal{C}', \mathcal{S}) = 
         \begin{cases}
            1 & \mbox{ if } \xi^u_{\mathcal{C}, \mathcal{S}} =1 \land \xi^u_{\mathcal{C}', \mathcal{S}} = 0, \\
            0 & \mbox{ otherwise},
        \end{cases}
    \end{align}
    where $\xi^u_{\mathcal{C}, \mathcal{S}}$ is defined by (\ref{xi_u_CS}).
\end{lemma}
\begin{proof}
    \begin{align}
        R(\mathcal{M}_\mathbb{P}, \mathcal{C}, \mathcal{S}) -  R(\mathcal{M}_\mathbb{P}, \mathcal{C}', \mathcal{S})
        & = \int_{\mathcal{V}} I'(u, \mathcal{C}, \mathcal{S}) - I'(u, \mathcal{C}', \mathcal{S})\mathrm{d}\mathbb{P}(u), \\
        \intertext{and $I'(u, \mathcal{C}, \mathcal{S}) - I'(u, \mathcal{C}', \mathcal{S})$ is nonnegative for all $u \in \mathcal{V}$ since $\mathcal{C}' \subseteq \mathcal{C}$. Thus,}
        & = \int_{\mathcal{V}} \xi_{\mathrm{diff}}(u, \mathcal{C}, \mathcal{C}', \mathcal{S}) \mathrm{d}\mathbb{P}(u).
    \end{align}
\end{proof}

Based on Lemmas \ref{lemma_recall_optimality_max}, \ref{lemma_recall_optimal_necessary_sufficient}, and \ref{lemma_recall_optimality_diff}, we show that $\mathcal{C}^*$ defined by (\ref{C_star}) achieves the condition \textbf{(PC2)} in Problem \ref{problem:PACoptimal_diff}.
In the following, we abbreviate $\mathcal{S}_{N, \delta, \beta}$ as $\mathcal{S}$, for simplicity.
For any $\mathcal{C}' \subset \mathcal{C}^*$, we define the subset of sampled parameters $\mathcal{U}^{\mathcal{C}^*, \mathcal{C}', \mathcal{S}}_{N, 1}$ as
\begin{align}
\label{U_C*C'Stilde_N1}
    \mathcal{U}^{\mathcal{C}^*, \mathcal{C}', \mathcal{S}}_{N, 1} = \mathcal{U}^{\mathcal{C}^*, \mathcal{S}}_{N, 1} \setminus \mathcal{U}^{\mathcal{C}', \mathcal{S}}_{N, 1},
\end{align}
where $\mathcal{U}^{\mathcal{C}, \mathcal{S}}_{N, 1}$ is defined by (\ref{U_CS'_N1}). Then, we define $\mathcal{U}_{N, 0}^{\mathcal{C}^*, \mathcal{C}', \mathcal{S}}$ as $\mathcal{U}_N \setminus \mathcal{U}_{N,1}^{\mathcal{C}^*, \mathcal{C}', \mathcal{S}}$.
We consider the following SOP:
\begin{align}
\label{SCP_for_thm3}
    \max_{\xi}\quad & \xi \nonumber \\
    \mathrm{s. t.} \quad & \xi \in  \bigcap_{u \in \mathcal{U}'_{N}} J^u_{\mathcal{C}^*, \mathcal{C}', \mathcal{S}},
\end{align}
where ${J}^u_{\mathcal{C}^*, \mathcal{C}', \mathcal{S}} = [0, \xi_\mathrm{diff}(u, \mathcal{C}^*, \mathcal{C}', \mathcal{S})]$, $\mathcal{U}'_N \subset \mathcal{U}_N$, and $\xi_\mathrm{diff}$ and $\xi^u_{\mathcal{C}, \mathcal{S}}$ are defined by (\ref{xi_diff}) and (\ref{xi_u_CS}), respectively. We denote the solution of (\ref{SCP_for_thm3}) by $\xi^*$. Note that $\xi^* = \min_{u \in \mathcal{U}'_{N}} \xi_\mathrm{diff}(u, \mathcal{C}^*, \mathcal{C}', \mathcal{S})$. For simplicity, we denote $\xi_\mathrm{diff}(u, \mathcal{C}^*, \mathcal{C}', \mathcal{S})$ as $\xi^u_\mathrm{diff}$.
\begin{proposition}
\label{prop_certify_PC2}
    For any upMDP $\mathcal{M}_\mathbb{P}$, any $N > 0$, any $\beta \in (0,1)$, if $\mathcal{C}^* \neq \emptyset$, then, for any proper subset $\mathcal{C}' \subset \mathcal{C}^*$, we have 
    \begin{align}
    \label{recall_optimal_inequality_diff}
        \mathbb{P}^N \left\{  R(\mathcal{M}_\mathbb{P}, \mathcal{C}^*, \mathcal{S}) > R(\mathcal{M}_\mathbb{P}, \mathcal{C}', \mathcal{S}) \right\} > \beta,
    \end{align}
    where $\mathcal{C}^*$ is defined by (\ref{C_star}), and $\mathcal{S}$ is defined by (\ref{S_delta_beta}).
\end{proposition}
\begin{proof}[Proof of Proposition \ref{prop_certify_PC2}]
% Theorem2の証明の(2)と基本は同じ．全てのu in U^{C*, C', S}に対して，Lemma4からxi^u_C*の方は正の値で，Lemma5からxi^u_C'の方は負の値を取る．
For any $\mathcal{U}'_N \neq \emptyset$ in (\ref{SCP_for_thm3}), the intersection of all $J^u_{\mathcal{C}^*, \mathcal{C}', \mathcal{S}}$ with $u \in \mathcal{U}'_N$ is nonempty and there exists a unique solution $\xi^*$ to (\ref{SCP_for_thm3}).
Thus, for any $\varepsilon > 0$ and any a-priori number of discarded sampled parameters $q = |\mathcal{U}_N \setminus \mathcal{U}'_N| < N$, by Theorem \ref{thm_scenario_discarding}, we have
    \begin{align}
    \label{eq_recall_opt_PACbound_joint}
        \mathbb{P}^N \left\{ \mathbb{P}(\xi^u_\mathrm{diff} \geq \xi^* ) \geq 1- \varepsilon
         \right\}
         \geq 1 - \sum_{i=0}^{q} \binom{N}{i} \varepsilon^i (1 - \varepsilon)^{N - i}.
        % \geq \left( 1 - \sum_{i=0}^{q} \binom{N}{i} \varepsilon^i (1 - \varepsilon)^{N - i} \right)^2.
    \end{align}
    % Since $\xi^*_0$ and $\xi^*_1$ are independently determined each other, we have
    % \begin{align}
    % \label{eq_recall_opt_PACbound_joint}
        % \mathbb{P}^N \left\{ \mathbb{P}(\xi^u_0 \geq \xi^*_0 \land \xi^u_1 \leq \xi^*_1) \geq (1- \varepsilon)^2 \right\} 
        % \geq 1 - \binom{q+1}{q} \sum_{i=0}^{q+1} \binom{N}{i} \varepsilon^i (1 - \varepsilon)^{N - i}.
        % \geq \left( 1 - \sum_{i=0}^{q} \binom{N}{i} \varepsilon^i (1 - \varepsilon)^{N - i} \right)^2.
    % \end{align}
    We set $t = 1 - \varepsilon$ and equate the right-hand side of (\ref{eq_recall_opt_PACbound_joint}) to $1 - (1- \beta)/N$, where $\beta \in (0,1)$. Then, we have
    \begin{align}
    \label{q_bound_for_relativeR}
        \mathbb{P}^N \left\{ \mathbb{P}(\xi^u_\mathrm{diff} \geq \xi^* ) \geq t^*(q) \right\} \geq 1 - \frac{1 - \beta}{N},
    \end{align} 
    where $t^*(q)$ is the solution of
    \begin{align}
    \label{t_root_solution}
        \frac{1 - \beta}{N} = \sum_{i=0}^{q} \binom{N}{i} (1-t)^i t^{N - i},
        % \beta = \left( 1 - \sum_{i=0}^{q} \binom{N}{i} (1 - t)^i t^{N-i} \right)^2,
    \end{align}
    By Bool's inequality, we have
    \begin{align}
    \label{Bool_inequality_relativeR}
        \mathbb{P}^N \left\{ \cap_{q=1}^{N} \mathcal{E}_q
    \right\} & \geq 1 - \sum_{q=1}^{N} \mathbb{P}^N \left\{ \mathcal{E}_q^c \right\} \nonumber \\
      & \geq \beta,
    \end{align}
    where $\mathcal{E}_q$ denotes the event $\mathbb{P}(\xi^u_\mathrm{diff} \geq \xi^*) \geq t(q)$ and $\mathcal{E}^c_q$ is its complement, and clearly
    \begin{align}
    \label{inequality_N_all_relativeR}
        \mathbb{P}^N \left\{ \mathcal{E}_{|\mathcal{U}^{\mathcal{C}^*, \mathcal{C}', \mathcal{S}}_{N,0}|} \right\} 
        \geq \mathbb{P}^N \left\{ \cap_{q=1}^{N} \mathcal{E}_q \right\}.
    \end{align}
    Note that we can take $\mathcal{U}'_N = \mathcal{U}^{\mathcal{C}^*, \mathcal{C}', \mathcal{S}}_{N,1}$ when $q = |\mathcal{U}^{\mathcal{C}^*, \mathcal{C}', \mathcal{S}}_{N,0}|$, and $\mathcal{U}'_N \neq \emptyset$ by (\ref{C_star}). Then, we have $\xi^u_{\mathcal{C}^*, \mathcal{S}} = 1$ and $\xi^u_{\mathcal{C}', \mathcal{S}} = 0$ for all $u \in \mathcal{U}^{\mathcal{C}^*, \mathcal{C}', \mathcal{S}}_{N,1}$ by Lemma \ref{lemma_recall_optimal_necessary_sufficient} and hence $\xi^* = 1$. Thus, by (\ref{Bool_inequality_relativeR}) and (\ref{inequality_N_all_relativeR}), we have
    \begin{align}
        \mathbb{P}^N \left\{ \mathbb{P}(\xi^u_\mathrm{diff} =  1) \geq t(|\mathcal{U}^{\mathcal{C}^*, \mathcal{C}', \mathcal{S}}_{N,0}|) \right\} & \geq \beta,
        \intertext{and $t(|\mathcal{U}^{\mathcal{C}^*, \mathcal{C}', \mathcal{S}}_{N,0}|) > 0$ holds since $\beta < 1$ and $|\mathcal{U}^{\mathcal{C}^*, \mathcal{C}', \mathcal{S}}_{N, 0}| < N$, thus,}
        \mathbb{P}^N \left\{ \mathbb{P}(\xi^u_\mathrm{diff} =  1) > 0 \right\} & \geq \beta,
    \end{align}
    Therefore , by Lemma \ref{lemma_recall_optimality_diff}, we have (\ref{recall_optimal_inequality_diff}).
\end{proof}

\begin{proof}[Proof of Theorem \ref{thm_solution_certif}]
    By Theorems \ref{thm_certif} and \ref{thm_optimality}, $\eta_N$ defined by (\ref{eta}) and $\zeta_N$ defined by (\ref{zeta}) are the solutions to Problem \ref{problem:PACbound}. Moreover, $\mathcal{C}^*$ fulfills \textbf{(PC1)} in Problem \ref{problem:PACoptimal_diff} by (\ref{bound_tightness_Cstar}) in Lemma \ref{lemma_recall_optimality_max} and fulfills \textbf{(PC2)} in Problem \ref{problem:PACoptimal_diff} by Proposition \ref{prop_certify_PC2}. Thus, we have the claim.
\end{proof}

\section{Computational Complexity of Algorithm \ref{Alg_PG_SPREst}}
\label{appendix:complexity}
The computational complexity of Algorithm 1 is in 
\begin{align}
    \mathcal{O} \left(N \cdot \mathrm{poly}(|S|) + N^2 |S| + 2^{|S|} \log{ \frac{1}{\varepsilon}} \right),
\end{align}
where $N$ is the number of sampled parameters, $|S|$ denotes the number of states, $\mathrm{poly}(n)$ denotes a polynomial function of $n$, and $\varepsilon$ is the approximation error for computing $\eta_N$, that is, $\varepsilon$ denotes the length of each sub-interval when dividing the unit interval $[0,1]$ with equal length to conduct bisection search for $\eta_N$.

The breakdown of the computational complexity is as follows. The complexity of computing the solution $\mathcal{C}^*$ to Problem 3 is in $\mathcal{O}(N \cdot \mathrm{poly}(|S|) + N^2|S|)$. The complexity of computing the lower bounds $\eta_N$ for all members $C \in \mathcal{C}^*$ is in $\mathcal{O}(2^{|S|} \cdot \log{\frac{1}{\varepsilon}})$ because we compute $\eta_N$ for all subsets $C \in 2^S$ in the worst case, and the complexity of bisection search to approximately compute $\eta_N$ is in $\mathcal{O}(\log{\frac{1}{\varepsilon}})$. However, in general, the computation of $\eta_N$ for all $C \in \mathcal{C}^*$ is not intensive. This is because the number of obtained potential causes is much smaller than $2^{|S|}$ in practice. For example, the number of potential causes obtained in Example is at most three. The complexity of computing $\zeta_N$ is in $\mathcal{O}(1)$ because $\zeta_N$ is immediately given as $(1 - \beta)^{1/N}$ by Theorem 3.

In what follows, we intuitively describe the details for the complexity to obtain the solution $\mathcal{C}^*$.

\begin{itemize}
    \item In Lines from 1 to 7:\\
    The complexity in Lines from 1 to 7 is in $\mathcal{O}(N \cdot \mathrm{poly}(|S|))$.
The complexity of Algorithm 2 is in $\mathcal{O}(\mathrm{poly}(|S|))$ by Lemma 4 in \citep{baier2022probability}. Also, please see \citep{baier2008principles} if you are not familiar with temporal model checking and its complexity. The complexity of probabilistic verifications for reachability properties in MDPs is in $\mathcal{O}(\mathrm{poly}(|S|))$ and is described by Corollary 10.107 in \citep{baier2008principles}, on which the result of Lemma 4 in \citep{baier2022probability} relies. $S_N$ can be computed by $S_N = \{ c \in S \;|\; \eta_N(\{c\}, \beta) \geq \delta \}$. This is because, for any $C \subset S$ and any $l \in \mathbb{R}$, satisfaction of $\eta_N(\{c\}, \beta) \geq l$ is a necessary condition for $\eta_N(C, \beta) \geq l$ by Definition 1, and hence $S_N$ can be constructed by gathering all states which satisfy the condition $\eta_N(\{c\}, \beta) \geq l$. So, the complexity for computing $S_N$ is in $\mathcal{O}(N \cdot |S|)$. Obviously, the computational complexity of obtaining $\mathscr{C}_{S_N, u}$ using $\tau$ is in $\mathcal{O}(|S|)$. The complexity of computing $C^*_{S_N, u}$ is in $\mathcal{O}(2^{|\phi_1|} \cdot |S|)$, where $ \phi_1 = \mathscr{C}_{S_N, u} \mathrm{U} c $ and $|\phi_1|$ denote the number of operators in $\phi_1$ and hence $|\phi_1| = 1$. This is because the complexity of model checking for a transition system (directed graph) and an LTL formula $\phi$ is in $\mathcal{O}(2^{|\phi|} \cdot |S|)$ by Theorem 3.3 in [Shonoebelen2002]. Thus, the computational complexity for $C^*_{S_N, u}$ is in $\mathcal{O}(|S|)$.

Hence, the complexity for Lines from 1 to 7 is in $\mathcal{O}(N \cdot \mathrm{poly}(|S|) )$.

\item In Lines from 9 to 11:\\
The complexity in Lines from 9 to 11 is in $\mathcal{O}(N^2 \cdot |S|)$.
For each $u_i \in \mathcal{U}_N$, the complexity of checking the condition in (22) for any $u_j \in \mathcal{U}_N$ is in $\mathcal{O}(|S|)$ by Theorem 3.3 in \citep{schnoebelen2002complexity}. Hence, the complexity to obtain $\mathcal{U}^i_N$ is in $\mathcal{O}(N \cdot |S|)$. So, the complexity of computing $\mathcal{U}^i_N$ for all $u_i \in \mathcal{U}_N$ is in $\mathcal{O}(N^2 \cdot |S|)$.

So, the overall computational complexity in Lines from 9 to 11 is in $\mathcal{O}(N^2 \cdot |S|)$.
\end{itemize}
Thus, the complexity of computing solution $\mathcal{C}^*$ is in $\mathcal{O}\left(N \cdot \mathrm{poly}(|S|) + N^2 |S|\right)$.

\section{Interpretability}
\label{appendix:interpretability}
Obtained SPR causes are interpreted through a probability-raising perspective and recall optimality as follows: They are subsets of states whose traversal increases the probability of undesired behavior in a certain portion of the parameter space of the upMDP, Moreover, the obtained collection of them cover the undesired paths maximally in most of the parameter space of the upMDP.

We can also give an interpretation similar to causal interventional interpretability \citep{moraffah2020causal}. For example, considering the following differences in conditional probabilities: $$\max_{\pi} \mathrm{Pr}^{\pi}_{\mathcal{M}_u} (\Diamond E \;|\; \Diamond C) - \mathrm{Pr}^{\pi}_{\mathcal{M}_u} (\Diamond E \;|\; \neg \Diamond C),$$ and $$\min_{\pi} \mathrm{Pr}^{\pi}_{\mathcal{M}_u}(\Diamond E \;|\; \Diamond C) - \mathrm{Pr}^{\pi}_{\mathcal{M}_u}(\Diamond E \;|\; \neg \Diamond C).$$

Intuitively, the differences in probability, respectively, represent how much passage of the cause $C$ affects the probability values of undesired behavior pessimistically and optimistically. These can be seen as variants of "what if" scenarios usually considered in a causal analysis context \citep{moraffah2020causal}. However, the above comparisons are not for an actual path but for the overall behavior of the system. 

As for larger systems modeled by upMDPs, such as the number of states that they have can reach millions, a possible difficulty for interpretability is the case where we have massive potential SPR causes, and hence we may fall into the situation where we struggle to grasp how much each potential cause affect the undesired behavior. In this case, we take various approaches to enhance ease of understanding for humans as follows:

\begin{enumerate}
    \item We calculate the aforementioned differences in probability for the union of potential causes rather than each potential cause. Then, we can interpret the differences in probability as how much the passage of any potential causes affects the probability of undesired behavior.
    \item We pick up the obtained potential causes to be interpreted. For example, we can choose $k$ potential causes from all of them in order of the values of lower bounds $\eta_N$ for SPR cause probability. Then, we can interpret the differences in probability as how much the passage of each top-$k$ potential cause affects the probability of undesired behavior.
\end{enumerate}

\section{Another example to comprehend how (\ref{recall_optimal_samples}) works}
\label{appendix:other_example}

\begin{example}
\label{example3}
We consider the upMDP depicted in Fig. \ref{example1:upMDP} with the set of parameter variables $u = [p, q]^\top$, the set of actions $A=\{a, b\}$, and the terminal bad state $E = \{s_3\}$. $p$ and $q$ follow the uniform distributions over $[0.11, 0.51]$ and $[0.3, 0.7]$, respectively.
\begin{figure}[htbp]
\centering
    \includegraphics[width=0.5\linewidth]{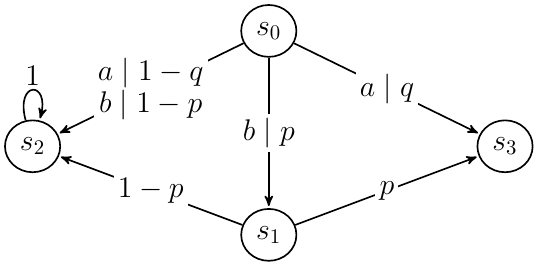}
\caption{The upMDP with the parameter variables $V=\{p, q\}$, the set of actions $A=\{ a, b\}$, and the terminal state $E = \{ s_3 \}$. The enabled action is only $a$ at states other than $s_0$ and thus we omit it.}
\label{example1:upMDP}
\end{figure}
\end{example}

In Fig. \ref{example1:upMDP}, we have no SPR cause when $q > p$. This is because there exists a policy $\pi$ such that $\pi(s_0, a) > (p - p^2 \pi(s_0, b) ) / q$, and hence we have 
\begin{align}
    \mathrm{Pr}^\pi_{\mathcal{M}_u} (\Diamond s_3) & = \pi(s_0, a) q + \pi(s_0, b) p^2 \nonumber \\
         & > p - p^2 \pi(s_0, b) + p^2 \pi(s_0, b) \nonumber \\
         & = \mathrm{Pr}^\pi_{\mathcal{M}_u}(\Diamond s_3 \;|\; \Diamond s_1)
\end{align}

However, when $p > q$, $\{s_1\}$ is the unique SPR cause. Let us explain the reason for this in two cases.

(1) When $q \geq p^2$, for any policy $\pi$, the probability of reaching $s_3$ from $s_0$ is $\mathrm{Pr}^\pi_{\mathcal{M}_u} (\Diamond s_3) = \pi(s_0, a) q + \pi(s_0, b) p^2$ and, since $q > p^2$, we have $\pi(s_0, a) q + \pi(s_0, b) p^2 < q$. So, we have $\mathrm{Pr}^\pi_{\mathcal{M}_u} (\Diamond s_3) < \mathrm{Pr}^\pi_{\mathcal{M}_u} (\Diamond s_3 \;|\; \Diamond s_1)$ since $q < p$. Thus, $\{s_1\}$ is the SPR cause.

(2) When $q < p^2$, the maximal probability of reaching $s_3$ from $s_1$ is $p^2$ and obviously $p^2 < p$. So, we have $\{s_1\}$ as the SPR cause.

We consider solving Problem \ref{problem:PACoptimal_diff} for the upMDP in Example \ref{example3}. Let $\delta=0.1$ and $\beta=0.99$. We sampled $N=1000$ parameters. We obtained a solution $\mathcal{C}^* = \{ \{s_1\} \}$ with $\mathcal{S}_{N, \delta, \beta}$ $ =\{ \{s_1\} \}$, where $\{s_1\}$ is the canonical SPR causes on the sampled $992$ MDPs with $p > q$. Obviously, $\mathcal{U}^i_N$ that corresponds to $\{s_1\}$ is equal to $\mathcal{U}_N$. According to Theorem \ref{thm_certif}, $\{s_1\}$ provides the lower bounds of SPR-cause probabilities $0.98$ with the confidence probability $\beta=0.99$, respectively, and they exceed $\delta$. For \textbf{(PC1)}, the lower bound $\zeta_N(\mathcal{C}^*, \mathcal{S}_{N, \delta, \beta}, \beta)$ of recall-optimal probability is the maximum value $(1 - \beta)^{1/N} = 0.995$. 
    For \textbf{(PC2)}, the subset $\mathcal{C} = \emptyset \subset \mathcal{C}^*$ clearly does not contain any SPR cause for all parameters. So, $R(\mathcal{M}_\mathbb{P}, \mathcal{C}^*, S_{N, \delta, \beta}) > R(\mathcal{M}_\mathbb{P}, \mathcal{C}, S_{N, \delta, \beta})$.

\section{Intuitive explanation for SPR causes on Example section}
\label{appendix:explain_example_sec}
The undesired behavior of the two environments depicted in \ref{ex:env1} and \ref{ex:env2} in Example is to enter the red cells. That is, the red cells are "undesired" terminal states. Our demonstration has aimed to nonredundantly and exhaustively find the causes of the undesired behavior. In this example, causes are seen as the states such that passing through the states raises the probability of the undesired behavior. Yellow, orange, and green cells correspond to the causes for some parameters of the upMDP. Our proposed method finds these potential causes by sampling upMDP parameters, model checking-based cause synthesis for each sampled MDP, and consolidating them by the proposed set covering technique.

We can observe that the three causes, depending on the parameter values in these environments, have been found well. Intuitively, in both environments, when the probability of entering the red cells from above is greater than that of entering the red cells from below, that is $p_1 < p_2$, the robot passing through the yellow cells is more likely to enter the red cells from above than other behavior that does not pass through the yellow cells. Likewise, when $p_1 > p_2$, the probability of entering the red cells from below rises after passing through the orange/green cells. The difference between the orange cells and green cells as the SPR causes is subtle. When $p_0$ is sufficiently large, that is $p_0$ is close to the upper limit $0.9$, the chance to unintentionally go up from the cell $(4, 5)$ is sufficiently small. Hence, the chance to visit the red cells from above is sufficiently small. So, in this case, the set of orange cells is an SPR cause. On the contrary, when $p_0$ is relatively small, that is $p_0$ is close to the lower limit $0.8$, the set of green cells is an SPR cause. Moreover, we can intuitively find that these three potential causes cover all undesired paths, and removing any potential cause makes an uncovered undesired path appear.

Our algorithm has effectively found these intuitive potential causes by combining three ingredients: (i) sampling parameters of the environment and PAC-bound we introduced, (ii) model checking-based SPR cause synthesis for each sampled MDP based on \citep{baier2022probability}, and (iii) consolidation of these synthesized causes by the proposed set covering based on (22). In detail, when we have tried to obtain subsets of states that become SPR causes with positive probability, we have obtained all three SPR causes. On the other hand, when we have tried to obtain subsets of states that have become SPR causes with the probability threshold $\delta=0.1$, we obtain only yellow cells. This is because the probability of $p_1 > p_2$ is sufficiently small.

\section{Detailed Results for Example}
\label{appendix:detailed_results}
Fig. \ref{ex:result_env1} and Fig. \ref{ex:results_env2}  show the results obtained from our proposed method for two environments depicted in Figs. \ref{ex:env1} and \ref{ex:env2}, respectively, with all experimental parameters $\beta = \beta_0, \beta_1$ and $\delta = \delta_0, \delta_1$. We depict the results with blue, green, magenta, and red lines for the parameters $(\delta_0, \beta_0)$, $(\delta_0, \beta_1)$, $(\delta_1, \beta_0)$, and $(\delta_1, \beta_1)$, respectively.
The top row in Fig. \ref{ex:results} shows the recall-optimal probabilities $R$, their lower bounds $\zeta$, and the minimum recall-optimal probabilities over all proper subsets of the obtained solution $\mathcal{C}^*$ under each sample size $N$ as dotted, solid, and dashed lines, respectively. We find that $\zeta$ and $R_\mathrm{sub}$ bound $R$ from below for all sample sizes and parameters. Moreover, $\zeta$ takes a higher value when $\beta = \beta_0$ than when $\beta = \beta_1$.
The middle row (resp., the bottom row) shows the maximum (resp., minimum) SPR-cause probabilities $F_\mathrm{max}$ (resp., $F_\mathrm{min}$) and the maximum (resp., minimum) lower bounds $\eta_\mathrm{max}$ (resp., $\eta_\mathrm{min}$) for all state sets $C$ in the obtained solution $\mathcal{C}^*$ as dotted and solid lines, respectively, for each sample size $N$. We observe that $\eta_\mathrm{max}$ (resp., $\eta_\mathrm{min}$) bounds $F_\mathrm{max}$ (resp., $\eta_\mathrm{min}$) from below and exceeds $\delta$ for all sample sizes and all experimental parameters.

\begin{figure*}[htbp]
  	\centering
  	\subfigure[]{
  		\includegraphics[width=0.48\linewidth]{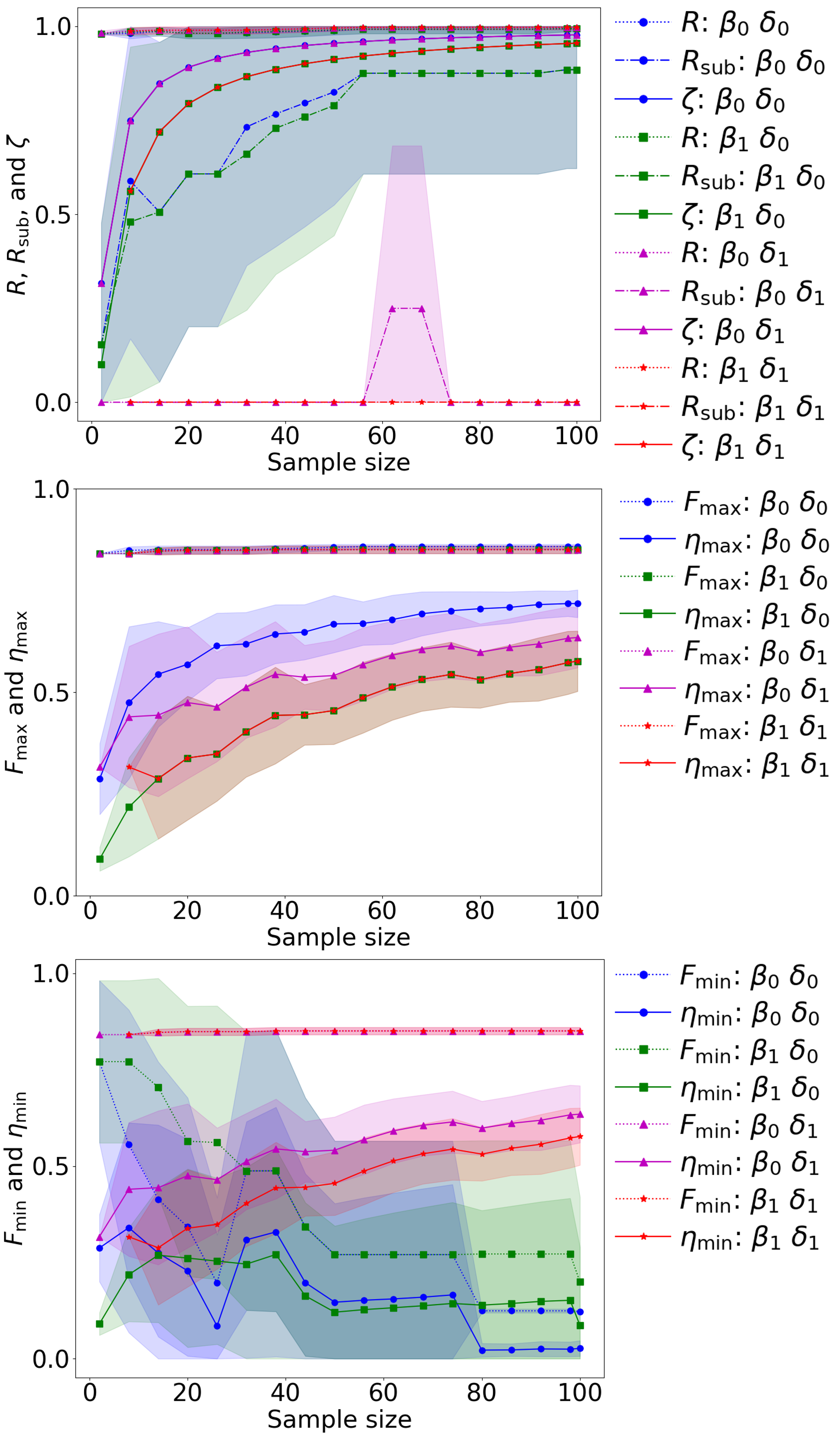}
            \label{ex:result_env1}
  	}
  	\subfigure[]{
  		\includegraphics[width=0.48\linewidth]{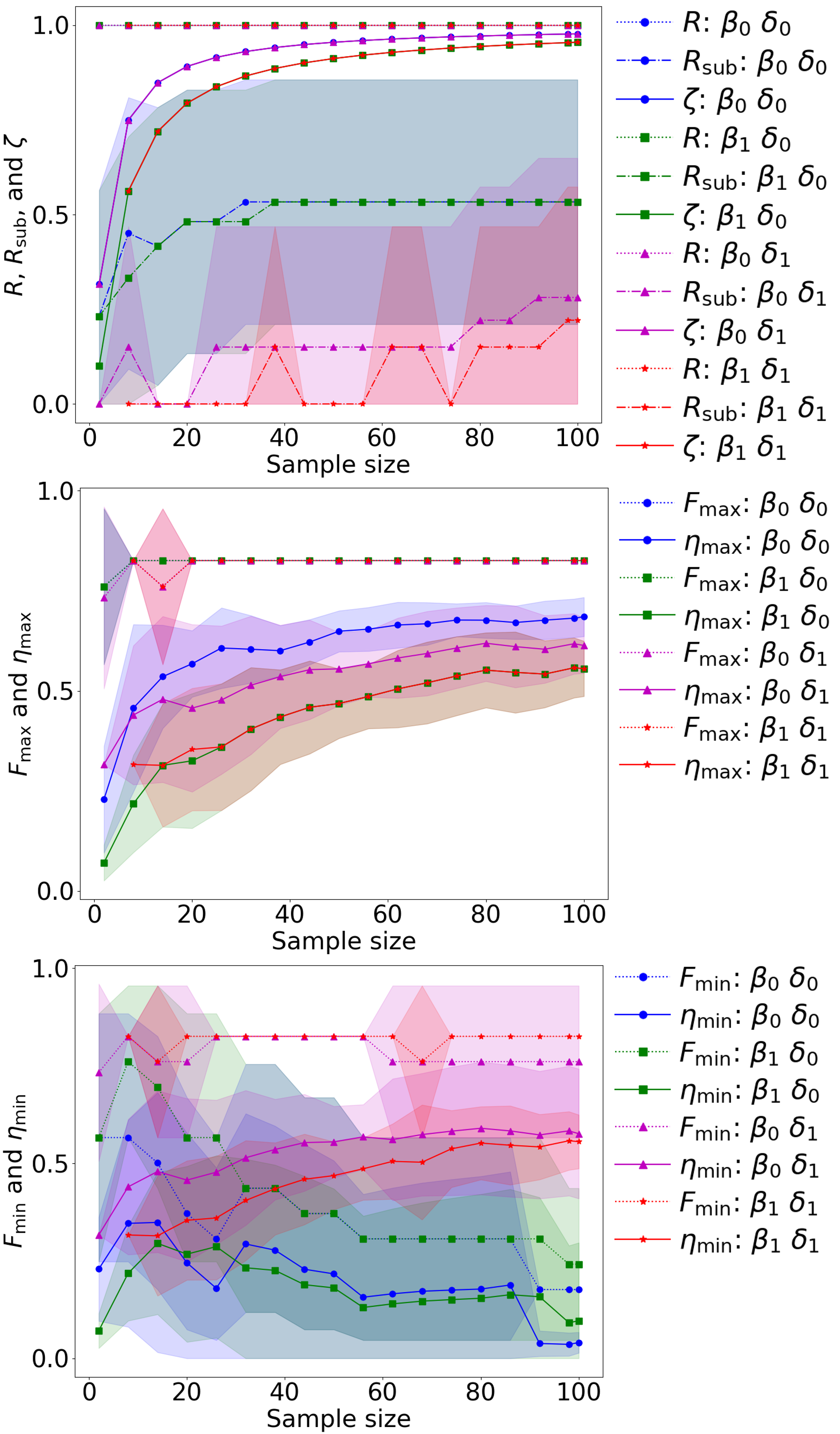}
  		\label{ex:results_env2}
  	}
  	\caption{(a) and (b) show the results with all experimental parameters for two environments depicted in Fig. \ref{ex:env1} and Fig. \ref{ex:env2}, respectively. (Top row): The means and standard deviations of the recall-optimal probabilities $R$ and their lower bounds $\zeta$ for the solutions obtained from the proposed method against each sample size (blue and magenta). We further show the maximum recall-optimal probabilities $R_\mathrm{sub}$ for all proper subsets of the solutions (green). (Middle row) (resp., (Bottom row)): The means and standard deviations of the maximum (resp., minimum) SPR-cause probabilities and their lower bounds over all potential SPR causes in the solutions obtained from the proposed method against each sample size. The results of the upper and the lower figures in (a)-(c) correspond to two environments depicted in Fig\ \ref{ex:env1} and \ref{ex:env2}, respectively.}
    \label{ex:results}
   \vspace{0mm}
\end{figure*}

\end{document}